\documentclass[12pt,a4paper]{article}

\usepackage{fourier}
\usepackage{microtype}
\usepackage[hmargin=2.5cm, vmargin=3.0cm]{geometry}
\usepackage{graphicx}
\usepackage{tikz}
\usepackage{tikz-qtree}
\usepgfmodule{matrix}
\usetikzlibrary{shapes,decorations}
\usepackage{amssymb,amsmath}
\usepackage{amsthm} 
\usepackage{dsfont}
\usepackage[official]{eurosym}
\usepackage{natbib}
\usepackage[lowtilde]{url} 
\usepackage{bm}
\usepackage{verbatim}
\usepackage{appendix}
\usepackage[normalem]{ulem}
\usepackage{cancel}

\numberwithin{equation}{section}

\theoremstyle{definition}

\theoremstyle{plain}
\newtheorem{theorem}{Theorem}[section]
\newtheorem{lemma}{Lemma}[section]

\newtheorem{assumption}{Assumption}[section]

\theoremstyle{remark} 
\newtheorem{remark}{Remark}[section]

\newtheorem{examplenorm}{Example}[section]

\newcommand{\HRule}{\rule{\linewidth}{0.5mm}}

\providecommand{\abs}[1]{\left\lvert#1\right\rvert} 

\newcommand{\EE}{\mathds{E}}  

\newcommand{\eFF}{\mathcal{F}}
\newcommand{\eBB}{\mathcal{B}}

\newcommand{\Ind}{\mathds{1}} 
\newcommand{\NN}{\mathds{N}}

\newcommand{\PP}{\mathds{P}}  
\newcommand{\RR}{\mathds{R}}  
\newcommand{\QQ}{\mathds{Q}}  
\newcommand{\DD}{\mathds{D}}

\newcommand{\Var}{\mathds{V}\kern-2pt\text{ar}} 
\newcommand{\intd}{\mathrm{d}}

\newcommand{\R}{\mathbb{R}}

\newcommand{\vari}{\mathbb{V}{\rm ar}}

\let\oldsqrt\sqrt
\def\sqrt{\mathpalette\APsqrt}
\def\APsqrt#1#2{%
\setbox0=\hbox{$#1\oldsqrt{#2}$}\dimen0=\ht0%
\advance\dimen0-0.2\ht0%
\setbox2=\hbox{\vrule height\ht0 depth -\dimen0}%
{\box0\lower0.48pt\box2}} 

\begin{document}

\begin{titlepage}
\begin{center}

\HRule \\[0.4cm]
{ \LARGE Fast Convergence of Regress-Later Estimates in Least Squares Monte Carlo}\\
\HRule \\[0.4cm]
\end{center}

\begin{center}
	\large Eric Beutner \footnotemark[1] \\
	\textit{Maastricht University}
\end{center}
\begin{center}
	\large Antoon Pelsser \footnotemark[2] \\
	\textit{Maastricht University \& Netspar \& Kleynen Consultants}
\end{center}
\begin{center}
	\large Janina Schweizer  \footnotemark[3]\\
	\textit{Maastricht University \& Netspar}
\end{center}


\footnotetext[1]{Maastricht University, Dept. of Quantitative Economics, P.O.~Box 616, 6200 MD Maastricht, The Netherlands, e.beutner@maastrichtuniversity.nl}
\footnotetext[2]{Maastricht University, Dept. of Quantitative Economics and Dept. of Finance, P.O.~Box 616, 6200 MD Maastricht, The Netherlands, a.pelsser@maastrichtuniversity.nl}
\footnotetext[3]{Maastricht University, Dept. of Quantitative Economics, P.O.~Box 616, 6200 MD Maastricht, The Netherlands, j.schweizer@maastrichtuniversity.nl}

\end{titlepage}

\begin{abstract}
Many problems in financial engineering involve the estimation of unknown conditional expectations across a time interval. 
Often Least Squares Monte Carlo techniques are used for the estimation. One method that can be combined with Least Squares Monte Carlo is  the ``Regress-Later'' method. Unlike conventional methods where the value function is regressed on a set of basis functions valued at the beginning of the interval, the ``Regress-Later'' method regresses the value function on a set of basis functions valued at the end of the interval. The conditional expectation across the interval is then computed exactly for each basis function. We provide sufficient conditions under which we derive the convergence rate of Regress-Later estimators. Importantly, our results hold on non-compact sets. We show that the Regress-Later method is capable of converging significantly faster than conventional methods and provide an explicit example. Achieving faster convergence speed provides a strong motivation for using Regress-Later methods in estimating conditional expectations across time.
\end{abstract}

\section{Introduction}\label{sec:Intro}
The Least Squares Monte Carlo (LSMC) technique is widely applied in the area of Finance to estimate conditional expectations across a time interval. Under LSMC the cross-sectional information inherent in the simulated data is exploited to obtain approximating functions to conditional expectations through performing least squares regressions on the simulated data. Examples may be found in \citet{Carriere_LSMC}, \citet{Broadie_Glasserman_American}, \citet{Longstaff_Schwartz}, \citet{Tsitsiklis_American}, \citet{Clement}, \citet{Stentoft_LSMC}, \citet{Glasserman_Yu_2004}, \citet{Egoff}, \citet{Belomestny}, \citet{Gerhold} and \citet{Zanger}, who discuss approaches to LSMC with application to American\slash Bermudan option pricing; see also \citet{Broadie_Glasserman_American} who apply simulation based methods and a dynamic programming algorithm to American option pricing.
These papers have in common that the conditional expectation at time $t$ is approximated through least squares regression of the value function at a time point $T>t$ against basis functions at the earlier time point $t$. This approach to the estimation of conditional expectations has been termed ``regression now'' by \citet{Glasserman_Yu_2002}. Here we will use the expression Regress-Now. In the same paper \citet{Glasserman_Yu_2002} introduce an alternative approach that they called ``regression later'' (throughout this paper Regress-Later). In Regress-Later the value function at a time point $T$ is approximated through LSMC techniques by basis functions that are measurable with respect to the information available at time $T$. Moreover, the basis functions in Regress-Later are selected such that the conditional expectation can be computed exactly. The conditional expectation of the time $T$ value function is then derived by evaluating the basis functions contained in the approximation function.
In this paper, we shall show that the Regress-Later method is fundamentally different from the Regress-Now technique. But before we briefly review recent contributions to the literature.\\
\citet{Glasserman_Yu_2002} show that the Regress-Later approach offers advantages compared to Regress-Now. They compare the properties of the coefficient estimates given that both approximations yield a linear combination of the same basis functions. Their results suggest that in a single-period problem the Regress-Later algorithm yields a higher coefficient of determination 
and a lower covariance matrix for the estimated coefficients; see also \citet{Broadie_Cao} who report similar observations. This implies that with Regress-Later potentially a better fit is attained with more accurate coefficient estimates. The results depend on more restrictive conditions on the basis functions as these are required to fulfill the martingale property. However, for many financial applications it is reasonable to expect that such a basis exists.  \citet{Bender_Steiner} use LSMC to numerically approximate the conditional expectations involved in estimating backward stochastic differential equations. They consider the Regress-Later algorithm and combine it with martingale basis functions as suggested in \citet{Glasserman_Yu_2002}. Their empirical case studies suggest that Regress-Later with martingale basis functions achieves a better numerical approximation at lower computational costs compared to traditional LSMC. The empirical results all show faster convergence rates for the Regress-Later algorithm combined with martingale basis functions compared to traditional LSMC.

Here, we shall shed more light on the advantages offered by Regress-Later as observed in \citet{Glasserman_Yu_2002}, \citet{Broadie_Cao} and \citet{Bender_Steiner} by analyzing the properties of Regress-Later in terms of its convergence rate. As it seems to be the first attempt to derive convergence rates for Regress-Later estimators we restrict ourselves to single-period problems. Our analysis will reveal that, as mentioned above, Regress-Later is fundamentally different from Regress-Now. Firstly, because Regress-Later can and does achieve a convergence rate for the mean-square error that is faster than $N^{-1}$; cf.~Section \ref{sec:convergence_regresslater}. Here and throughout $N$ is the sample size. We shall present an example where the convergence is indeed much faster than $N^{-1}$; cf.~Section \ref{sec:piecewise linear}. This is in sharp contrast to Regress-Now that can never converge faster than $N^{-1}$. We provide explanations for both facts, i.e.~the bound $N^{-1}$ for Regress-Now and the faster convergence rate for Regress-Later. It will turn out that the latter is a consequence of the fact that Regress-Later is a non-standard regression problem, because the variance of the noise term converges to zero. Secondly, we shall explain that the conditions needed to derive convergence rates for Regress-Later estimators are much weaker than the typical assumptions used in the literature for Regress-Now estimators; an exception is the recent work by \citet{Zanger}. This has to do with the fact that for Regress-Now estimators reasonable conditions stemming from nonparametric statistics were employed in the literature whereas for Regress-Later estimators we should definitely use parametric assumptions. Thereby, we will easily obtain approximations of the value function on non-compact intervals; see the discussion in Section \ref{sec:convergence_regresslater}.
Apart from these fundamental differences we will also present several examples which show that the functions to be approximated in Regress-Now may differ in nature compared to Regress-Later. Furthermore, we explain that the nonparametric assumptions that were applied in deriving convergence rates for Regress-Now estimators (see, for instance, \citet{Stentoft_LSMC}) can be weakened. These relaxed assumptions allow us to approximate the value function on the entire real line by a Regress-Now estimator and not only on a compact domain.

The structure of this paper is as follows. Section \ref{sec:Hilbert_Sieve} introduces the general LSMC estimator with sieve and distinguishes between its Regress-Now and Regress-Later applications. In Section \ref{sec:convergence_regresslater} the asymptotic convergence rate for Regress-Later estimators is derived under conditions that allow to approximate the value function on non-compact intervals. Moreover, similar conditions are applied to Regress-Now estimators while a motivation is given for when these conditions may be applicable for the Regress-Now technique. We conclude this section by providing explanations for the different convergence rates of Regress-Now and Regress-Later estimators. Section \ref{sec:piecewise linear} introduces an orthonormal basis based on piecewise linear functions and derives the explicit convergence rate for Regress-Later with that basis. Section \ref{sec:Conclusion} concludes. The proofs of all auxiliary results are presented in the appendix.


\section{Mathematical Model for Regress-Now and Regress-Later}\label{sec:Hilbert_Sieve}
As described in the introduction Regress-Now and Regress-Later are simulation based techniques to estimate conditional expectations.
 Often they are combined with series or sieve estimation, where the number of regressors in the least squares estimation is not fixed and finite; for an overview on series and sieve estimation one may refer to \citet{Chen_Sieve}.
In this section, we describe the mathematical model that is used throughout and explain the Regress-Now and Regress-Later approaches within this model.




We start with our mathematical model. Let $Z=\{Z(t), 0
\leq t \leq T\}$ be a $d$-dimensional stochastic process with $d \in
\mathds{N}$ defined on some filtered probability space $(\Omega,\eFF,\lbrace \eFF_t \rbrace_{0 \leq t \leq T}, \tilde{\PP})$. We denote the filtration generated by $Z$ by $\{\mathcal{F}_t\}_{0 \leq t \leq T}$. The measure $\tilde{\PP}$ denotes some probability measure equivalent to the true probability measure $\PP$. We leave $\tilde{\PP}$ generally unspecified when developing our model, but  provide the reader with an interpretation of the mathematical model for selecting $\tilde{\PP}$ just before Subsection \ref{sec:regress-now}.
The paths $Z(\cdot,\omega)$ of $Z$ given by $t \rightarrow Z(t,\omega)$, $t \in [0,T]$, are assumed to lie in some function space $\DD_d[0,T]$
consisting of functions mapping from $[0,T]$ to $\RR^d$, and we consider $Z$ as a random function. If $d=1$ we just write $\DD[0,T]$ and $\RR$.
We assume that the payoff $X$ is $\mathcal{F}_T$-measurable and that for every $\omega$ in
the sample space $\Omega$ the payoff $X(\omega)$ of the contingent claim $X$
can be written as $g_T(A_T(Z(\cdot,\omega)))$, where $A_T$ is a
known (measurable) functional mapping from $\DD_d[0,T]$ to $\RR^{\ell}$ and $g_T$ is
a known Borel-measurable function that maps from $\RR^{\ell}$ to $\RR$. This basically means that the payoff function $X$ depends only on finitely many characteristics of the stochastic paths of the underlying process. These characteristics are comprised in the functional mapping $A_T$.
The notation is very powerful for our purposes later on, and we illustrate it here with an example.

\begin{examplenorm}\label{example-asian-opt}{\it (Asian option)} Let $Z$ be one-dimensional and $X=(\int_0^T Z_1(u)\,\intd u-K)^+$, where $K$ is the strike price. Then $X$ does only depend on $\int_0^T Z_1(u)\,\intd u$. Thus, $A_T(f)=\int_0^T f(u)\,\intd u$ for every function $f \in \DD[0,T]$ and therefore $\ell=1$.
\end{examplenorm}

\noindent Observe that at time $T$ it suffices to observe the time average of the stochastic process rather than the whole path. Further examples that highlight the idea behind the notation are given in Sections \ref{sec:regress-now} and \ref{sec:regress-later}.

In the relevant literature, it has
become standard to restrict attention to square-integrable random
variables; \citep[see e.g.][]{Stentoft_LSMC, Bergstrom_Hilbert,MadandMil_1994, Longstaff_Schwartz}. We do the same here, that is we assume $g_T \in L_2\big(\RR^{\ell}, \eBB(\RR^{\ell}), \tilde{\PP}^{A_T(Z)}\big)$ (implying that $X$ is square-integrable, because $X(\omega)=g_T(A_T(Z(\cdot,\omega)))$)
where $\eBB(\RR^{\ell})$ denotes the Borel $\sigma$-algebra on $\RR^{\ell}$, and $\tilde{\PP}^{A_T(Z)}$ denotes the probability measure on $\RR^{\ell}$ induced
by the mapping $A_T(Z)$. Recall that $L_2\big(\RR^{\ell}, \eBB(\RR^{\ell}), \tilde{\PP}^{A_T(Z)}\big)$ is a Hilbert space with
inner product
\begin{eqnarray*}
\int_{\RR^{\ell}} h_1(u)h_2(u)\, \intd \tilde{\PP}^{A_T(Z)}(u)=\EE_{\tilde{\PP}}[h_1(A_T(Z))h_2(A_T(Z))]
\end{eqnarray*}
and norm
\begin{eqnarray*}
\sqrt{\int_{\RR^{\ell}} h_1(u)h_1(u)\, \intd \tilde{\PP}^{A_T(Z)}(u)}=\sqrt{\EE_{\tilde{\PP}}[h_1^2(A_T(Z))]}.
\end{eqnarray*}

As already mentioned the quantity of interest is $\EE_{\tilde{\PP}}[X|\mathcal{F}_t]$ where $\tilde{\PP}$ denotes a probability measure.
If we take $\tilde{\PP}=\QQ$ where $\QQ$ is the equivalent risk-neutral probability measure, then $\EE_{\QQ}[D(t,T)X|\mathcal{F}_t]$, where $D(t,T)$ is the discount factor for the period $t$ to $T$, corresponds to the time $t$ price of $X$. As a further example for the importance of $\EE_{\tilde{\PP}}[X|\mathcal{F}_t]$ take $\tilde{\PP}=\PP$, where $\PP$ is the true probability measure. Then, $\EE_{\tilde{\PP}}[X|\mathcal{F}_t]$ is the best $L_2-$approximation to $X$ that is measurable w.r.t.~the $\sigma$-field $\mathcal{F}_t$.
In Sections \ref{sec:convergence_regresslater} and \ref{sec:piecewise linear} we will use $\tilde{\PP}$ and leave it unspecified. Regress-Now with sieves and Regress-Later with sieves are two different simulation-based approaches to obtain an approximation to the time $t$ value of $X$. We outline both approaches in the following subsections.


\subsection{Regress-Now} \label{sec:regress-now}
We first describe the Regress-Now approach which is currently more popular.
To describe the Regress-Now approach, we assume that the quantity of interest, $\EE_{\tilde{\PP}}[X|\mathcal{F}_t]$, can be written as
\begin{equation*}\label{eq:regress_now_target}
g_{0,t}\big(A_t(Z)\big) = \EE_{\tilde{\PP}}\left[X|\mathcal{F}_t\right], \; 0\leq t < T,
\end{equation*}
where $A_t$ is a known (measurable) functional mapping from $\DD_d[0,t]$ to $\RR^{s}$ and $g_{0,t}$ is
an unknown Borel-measurable function that maps from $\RR^{s}$ to $\RR$. Here, $\DD_d[0,t]$ is the restriction of $\DD_d[0,T]$ to the interval $[0,t]$.

\begin{remark}
The notation $g_{0,t}(A_t(Z))$ is used to emphasize that the function $g_{0,t}$ is generally unknown. Thus, we use the convention that a subscript `0' indicates the true but unknown parameter. In contrast, note that $g_T(A_T(Z))$ refers to the payoff function, which is known in a simulation-based model as the simulation is controlled by the modeler.
\end{remark}

We give a few examples below for $g_{0,t}$ and $A_t$ that serve to illustrate the notation and concept. In these examples we take $\tilde{\PP}=\QQ$ to emphasise the pricing aspect of conditional expectations and for convenience we assume that the discount factor is equal to 1. 

\begin{examplenorm}\label{example-european-opt-now}{(\it European call with Regress-Now)} Let $Z$ be one-dimensional and consider an European call. Then $X=(Z_1(T)-K)^+$, where $K$ is the strike price. Moreover, $\EE_{\QQ}\left[X|\mathcal{F}_t\right]$ does only depend on $Z_1(t)$. Hence, we can take $A_t(f)=f(t)$ for every function $f \in \DD[0,t]$ and therefore $s=1$.
\end{examplenorm}

\begin{examplenorm}\label{example-european-dim-d-opt-now}{(\it European basket option with Regress-Now)} Consider a $d$-dimensional European basket option of the type $X=\big(\sum_{i=1}^d Z_i(T)-K\big)^+$, where $K$ is the strike price. In general $\EE_{\QQ}\left[X|\mathcal{F}_t \right]$ depends on $\mathbf{Z}(t)=(Z_1(t),\ldots,Z_d(t))$ and not only on $\sum_{i=1}^d Z_i(t)$. Then $A_t(f)=f(t)$ for every function $f \in \DD_d[0,t]$ and therefore $s=d$. We give an example that shows our claim. Consider two assets $Z_1(t)$ and $Z_2(t)$, $t=0,1,2$, that are independent under $\QQ$ with
\begin{eqnarray*}
& \QQ(Z_1(0)=10)=1; \hspace{0.3cm} \QQ(Z_1(1)=12)=\QQ(Z_1(1)=6)=0.5;\\
&  \hspace{0.2cm} \QQ(Z_1(2)=14|Z_1(1)=12)=\QQ(Z_1(2)=8|Z_1(1)=12)=0.5,\hspace{0.2cm} \QQ(Z_1(2)=6|Z_1(1)=6)=1,
\end{eqnarray*}
and
\begin{eqnarray*}
& \QQ(Z_2(0)=10)=1; \hspace{0.3cm} \QQ(Z_2(1)=12)=\QQ(Z_2(1)=6)=0.5;\\
 & \QQ(Z_2(2)=14|Z_2(1)=12)=\QQ(Z_2(2)=8|Z_2(1)=12)=0.5;\\
 & \QQ(Z_2(2)=9|Z_2(1)=6)=\QQ(Z_2(2)=1|Z_2(1)=6)=0.5.
\end{eqnarray*}



\noindent Take $X= \left( Z_1(2)+Z_2(2)- K \right)^{+}$ with $K=10$. We are interested in the conditional expectation at time $t=1$, i.e. $\EE_{\QQ}[X | \mathcal{F}_1]$. For the case where $Z_1(1)+Z_2(1)=18$ we obtain the following results
\begin{align*}
& \EE_{\QQ}[X |Z_1(1)= 12, Z_2(1)= 6 ] = 6.25 \mbox{ and } \EE_{\QQ}[X |Z_1(1)= 6, Z_2(1)=12 ] = 7.
\end{align*}
We immediately see that knowing the sum $Z_1(1)+Z_2(1)$ does not suffice to determine the conditional expectation at time $t=1$ as for $Z_1(1)+Z_2(1) = 18$ the conditional expectation can either be $6.25$ or $7$.
\end{examplenorm}

\begin{examplenorm}\label{example-asian-opt-now}{(\it Asian option with Regress-Now)} Let $Z$ be one-dimensional and $X=\big(\int_0^T Z_1(u)\,\intd u-K \big)^+$, where $K$ is again the strike price. Then $\EE_{\QQ}\left[X|\mathcal{F}_t \right]$ does only depend on $\int_0^t Z_1(u)\,\intd u$ and $Z_1(t)$. Thus, $A_t(f)= \left( \int_0^t f(u)\,\intd u,f(t) \right)$ for every function $f \in \DD[0,t]$ and therefore $s=2$.
\end{examplenorm}

\begin{examplenorm}\label{example-mild path-dependent-now}{(\it Mildly path-dependent option with Regress-Now)} Let $Z$ be one-dimensional, let $X$ be a function of $Z(u)$, $u<T$, and $Z(T)$, i.e.~$X=g_T(Z(u),Z(T))$ for some function $g_T$ and suppose that the expectation $\EE_{\QQ}\left[X|\mathcal{F}_t \right]$ depends only on $Z(t)$ for $t<u$. Then $A_t(f)=f(t)$ for every function $f \in \DD[0,t]$ and therefore $s=1$. 
\end{examplenorm}

\noindent The above examples illustrate the notation used for the Regress-Now model. 
We contrast the Regress-Now examples to their Regress-Later counterparts in Subsection \ref{sec:regress-later}.

In the following we describe how the Regress-Now {\it with sieves} approach estimates $g_{0,t}$. The description is rather detailed, because we will use it in Section \ref{sec:convergence_regresslater} to explain the different convergence rates for Regress-Now and Regress-Later.



Recall that the square-integrability of $X$ implies that $\EE_{\tilde{\PP}}[X|\mathcal{F}_t]$ is square-integrable as well. Hence, we also have that $g_{0,t} \in L_2\big(\RR^s, \eBB(\RR^s), \tilde{\PP}^{A_t(Z)}\big)$. Since the space $L_2\big(\RR^s, \eBB(\RR^s), \tilde{\PP}^{A_t(Z)}\big)$ is separable, $g_{0,t}$ is expressible in terms of a countable orthonormal basis $\left\lbrace e_k^{\text{\textit{now}}} \right\rbrace_{k=1}^{\infty}$
\begin{equation*}
    g_{0,t} = \sum_{k=1}^{\infty} \alpha_k^{\text{\textit{now}}} e_k^{\text{\textit{now}}};
\end{equation*}
see, for instance, \citet[Corollary 4.2.2~and Corollary 4.3.4]{Bogachev}. Because $g_{0,t}(A_t(Z))$ is the projection of $X$, the coefficients are given as
\begin{equation*}\label{equation coefficient g_{0,t}}
\alpha_k^{\text{\textit{now}}} = \EE_{\tilde{\PP}}[Xe_k^{\text{\textit{now}}}(A_t(Z))].
\end{equation*}
Thus, in particular, we have
\begin{equation}\label{eq:basis-rep}
    g_{0,t}\left(A_t(Z)\right) = \sum_{k=1}^{\infty} \alpha_k^{\text{\textit{now}}} e_k^{\text{\textit{now}}}\left(A_t(Z)\right);
\end{equation}
and, as usual, we define the projection error $p_{0,t}(A_T(Z))$ by
\begin{equation*}\label{equation projection error}
p_{0,t}(A_T(Z)):=X-g_{0,t}(A_t(Z))
\end{equation*}
which implies the well-known representation
\begin{equation*}\label{equation regress}
X=g_{0,t}(A_t(Z))+p_{0,t}(A_T(Z)).
\end{equation*}
Notice also that by construction $g_{0,t}(A_t(Z))$ and $p_{0,t}(A_T(Z))$ are orthogonal, i.e.
\begin{equation*}\label{equation ortho regress and projection}
\EE_{\tilde{\PP}}\left[g_{0,t}(A_t(Z))p_{0,t}(A_T(Z))\right]=0.
\end{equation*}
The Regress-Now approach tries to estimate the unknown function $g_{0,t}$ through its representation in Equation \eqref{eq:basis-rep} by generating data under $\tilde{\PP}$. However, Equation \eqref{eq:basis-rep} involves infinitely many parameters, which leaves the estimation infeasible. Sieve estimation offers solutions by estimating the model through finite-dimensional representations, which grow in complexity as the sample size increases and thereby yield the true outcome in the limit. For Equation \eqref{eq:basis-rep} this implies that with sieves we approximate $g_{0,t}$ by
\begin{equation*}
    g_{0,t}^K := \sum_{k=1}^K \alpha_k^{\text{\textit{now}}} e_k^{\text{\textit{now}}}= \left(\boldsymbol{\alpha}^{\text{\textit{now}}}_K\right)^T \bm{e}_K^{\text{\textit{now}}},
\end{equation*}
where $\boldsymbol{\alpha}^{\text{\textit{now}}}_K=\left( \alpha_1,\ldots,\alpha_K \right)^T$, $\bm{e}_K^{\text{\textit{now}}}=\left( e_1^{\text{\textit{now}}},\ldots,e_K^{\text{\textit{now}}} \right)^T$, and $T$ denotes transpose. Thus, a superscript ${}^T$ means transpose and it should be easy to distinguish it from the terminal time $T$. This results in an approximation error $a_{0,t}^K$ for $g_{0,t}$ given by
\begin{equation*}\label{def of approximation error}
a_{0,t}^K:=g_{0,t}-g_{0,t}^K.
\end{equation*}
Notice that we have $\EE_{\tilde{\PP}}\left[ g_{0,t}^K(A_t(Z))a_{0,t}^K(A_t(Z)) \right]=0$ by construction. By definition the approximation error $a_{0,t}^K$ converges to zero as $K \to \infty$ in $L_2$. Moreover, we can now write $X$ as
\begin{equation*}\label{eq:regresseq-now}
X=g_{0,t}^K(A_t(Z))+a_{0,t}^K(A_t(Z))+p_{0,t}(A_T(Z)).
\end{equation*}
From the last equation we can clearly see that the difference between $X$ and $g_{0,t}^K(A_t(Z))$ results from two sources: an approximation error and a projection error.\\

Now, given a (simulated) sample of size $N$ denoted by $\big((x_1,A_t(z_1)),\ldots,(x_N,A_t(z_N))\big)$ it is natural to estimate $g_{0,t}^K$ by the `sample projection'
\begin{equation*}
    \hat{g}_{0,t}^K = \arg \min_{g \in \mathcal{H}_K} \frac{1}{N} \sum_{n=1}^N \left(x_n - g(A_t(z_n)) \right)^2,
\end{equation*}
where $\mathcal{H}_K:=\left\{g: \RR^s \rightarrow \RR \;| \; g=\sum_{k=1}^K \alpha_k e_k^{\text{\textit{now}}}, \alpha_k \in \RR \right\}$.
Thus, we have
\begin{equation*}\label{eq:now_estimator}
   \hat{g}_{0,t}^K = \left(\bm{\hat{\alpha}}_K^{\text{\textit{now}}}\right)^T\mathbf{e}_K^{\text{\textit{now}}},
\end{equation*}
with
\begin{align*}
  {\hat{\boldsymbol{\alpha}}}_K^{\text{\textit{now}}} = \left(\left(\bm{E}_K^{{\text{\textit{now}}}}\right)^T\bm{E}_K^{{\text{\textit{now}}}}\right)^{-1} \left(\bm{E}_K^{\text{\textit{now}}}\right)^T\bm{X}, \end{align*}
where $\bm{X} = \left(x_{1}, ..., x_{N}\right)^T$ and $\bm{E}_K^{\text{\textit{now}}}$ is an $N \times K$ matrix with the $n^{\text{th}}$ row equal to $\bm{e}_K^{\text{\textit{now}}}(A_t(z_{n}))$, $n=1,\ldots,N$.
Notice that ${\hat{\boldsymbol{\alpha}}}_K^{\text{\textit{now}}}$ corresponds to the usual least squares estimator from a regression of $X$ against $K$ basis functions valued at time $t$.

\subsection{Regress-Later} \label{sec:regress-later}
In the previous section we discussed the Regress-Now approach.
Regress-Later proceeds as follows to approximate the quantity of interest, i.e.~$\EE_{\tilde{\PP}}\left[X|\mathcal{F}_t\right],  0 \leq t < T$. Approximate first the payoff $X$ by basis functions, mathematically speaking random variables, for which the calculation of the conditional expectation at time $t$ is exact. Then, given the linear representation of $X$ through basis functions, apply the operator $\EE_{\tilde{\PP}}\left[\cdot|\mathcal{F}_t\right]$ to these basis functions. The approach takes advantage of the linearity of the expectation operator. Note that the two-step approach is advantageous if for the payoff function $X$ basis functions exist that can easily be evaluated under the conditional expectation. For the case where $\tilde{\PP} = \QQ$ this implies that closed-form solutions for the price of the basis function must be readily available. We introduce a very simple but effective basis function in Section \ref{sec:piecewise linear}.
We now describe the Regress-Later approach and address differences to the Regress-Now approach.

Recall that we assume $X=g_T(A_T(Z))$ with $A_T$ a
known (measurable) functional mapping from $\DD_d[0,T]$ to $\RR^{\ell}$ and $g_T$ a known Borel-measurable function that maps from $\RR^{\ell}$ to $\RR$.
The examples below shall illustrate the meaning of $g_T$ and $A_T$. They may also be compared to their counterparts in Section \ref{sec:regress-now}.

\begin{examplenorm}\label{example-european-opt-lat}{\it (European call with Regress-Later)} Let $Z$ be one-dimensional and consider an European call. Then $X=(Z_1(T)-K)^+$, where $K$ is the strike price. Then, $X$ does only depend on $Z_1(T)$. Therefore, we can take $A_T(f)=f(T)$ for every function $f \in \DD[0,T]$ and hence $\ell=1$. Moreover, $g_T$ is given by $g_T(x)=(x-K)^+$.
\end{examplenorm}

\begin{examplenorm}\label{example-european-dim-d-opt-lat}{\it (European basket option with Regress-Later)} Consider the $d$-dimensional European basket option of Example \ref{example-european-dim-d-opt-now}. Then we can take $A_T(f)=\sum_{i=1}^d f_i(T)$ for every function $f \in \DD[0,T]$ and therefore $\ell=1$. Compare with Example \ref{example-european-dim-d-opt-now} where we had $s=d$. Moreover, $g_T$ is again given by $g_T(x)=(x-K)^+$ with $K$ the strike price.
\end{examplenorm}

\begin{examplenorm}\label{example-asian-opt-lat}{\it (Asian option with Regress-Later)} This corresponds to Example \ref{example-asian-opt}. For the readers' convenience and to contrast it with Example \ref{example-asian-opt-now} we repeat it here. Let $Z$ be one-dimensional and $X=\big(\int_0^T Z_1(u)\,\intd u-K \big)^+$, where $K$ is the strike price. Then $X$ does only depend on $\int_0^T Z_1(u)\,\intd u$. Thus, $A_T(f)=\int_0^T f(u)\,\intd u$ for every function $f \in \DD[0,T]$ and therefore $\ell=1$. Compare with Example \ref{example-asian-opt-now} where we had $s=2$. Again, we have $g_T(x) = (x-K)^{+}$.
\end{examplenorm}

\begin{examplenorm}\label{example-mild path-dependent-lat}{\it (Mildly path-dependent option with Regress-Later)} Let $X$ be as in Example \ref{example-mild path-dependent-now}. Then $A_T(f)=(f(u),f(T))$ for every function $f \in \DD[0,T]$ and therefore $\ell=2$. Recall that we had $s=1$ in Example \ref{example-mild path-dependent-now}. 
\end{examplenorm}

\noindent The above examples illustrate the notation that is applied to the Regress-Later approach. They also show that there may be fundamental differences between Regress-Now and Regress-Later. As already mentioned in the introduction the functions to be approximated in Regress-Now may differ in nature compared to Regress-Later. Recall that in Regress-Now the unknown function $g_{0,t}(A_t(Z))$ is approximated while in Regress-Later the known function $g_T(A_T(Z))$ is initially of interest. Although the ultimate goal in both Regress-Now and Regress-Later is to approximate the conditional expectation function, the approximation approaches are very different. As Examples \ref{example-european-dim-d-opt-lat}, \ref{example-asian-opt-lat}, \ref{example-mild path-dependent-lat} and their Regress-Now counterparts show the dimensionality of the function to be approximated under the same problem set-up may differ between Regress-Now and Regress-Later. The dimensionality of the function to be approximated may be one decision criterion in choosing between Regress-Now and Regress-Later. In Section \ref{sec:convergence_regresslater} we investigate the speed of convergence of both estimators and provide a strong argument for using Regress-Later estimators.

In the following we describe the Regress-Later estimation {\it with sieves}.
Remember that we assume square-integrability of the payoff function, meaning that $g_T \in L_2\big(\RR^{\ell}, \eBB(\RR^{\ell}), \tilde{\PP}^{A_T(Z)}\big)$. Hence, by the same argument as in Section \ref{sec:regress-now},
\begin{equation*}\label{eqaution X as infintie sum in l^2}
    X=g_{T}(A_T(Z)) = \sum_{k=1}^{\infty} \alpha_k^{\text{\textit{lat}}} e_k^{\text{\textit{lat}}}(A_T(Z)),
\end{equation*}
where $\left\lbrace e_k^{\text{\textit{lat}}} \right\rbrace_{k=1}^{\infty}$ is a countable orthonormal basis of $L_2\big(\RR^{\ell}, \eBB(\RR^{\ell}), \tilde{\PP}^{A_T(Z)}\big)$. Notice that by construction the projection error is zero in contrast to the Regress-Now approach. The coefficients $\alpha_k^{\text{\textit{lat}}}$ are given by
\begin{equation*}\label{alpha lat}
\alpha_k^{\text{\textit{lat}}}=\EE_{\tilde{\PP}}\left[X e_k^{\text{\textit{lat}}}(A_T(Z))\right].
\end{equation*}
As in Regress-Now approaches we apply sieves and approximate
$$g_T=\sum_{k=1}^{\infty} \alpha_k^{\text{\textit{lat}}} e_k^{\text{\textit{lat}}}$$
by a finite number of regressors, i.e.
$$g_T^K=\sum_{k=1}^{K} \alpha_k^{\text{\textit{lat}}} e_k^{\text{\textit{lat}}}=\left(\boldsymbol{\alpha}_K^{\text{\textit{lat}}}\right)^T\bm{e}_K^{\text{\textit{lat}}},$$
where $\boldsymbol{\alpha}_K^{\text{\textit{lat}}}=(\alpha_1^{\text{\textit{lat}}},\ldots,\alpha_K^{\text{\textit{lat}}})^T$ and $\bm{e}_K^{\text{\textit{lat}}}=(e_1^{\text{\textit{lat}}},\ldots,e_K^{\text{\textit{lat}}})^T$.
Defining the approximation error $a_T^K$ as usual by $a_T^K:=g_T-g_T^K$ we obtain the representation
\begin{equation}\label{eq:regresseq-lat}
X=g_T^K(A_T(Z))+a_T^K(A_T(Z)),
\end{equation}
which, as already mentioned, does not contain a projection error. Notice also that\\ $\EE_{\tilde{\PP}}\left[g_T^K(A_T(Z))a_T^K(A_T(Z))\right]=0$. It should be emphasized again that the approximation error converges to zero as $K \to \infty$ in $L_2$. As for Regress-Now with sieves given a (simulated) sample of size $N$ denoted by $(x_1,A_T(z_1)),\ldots,(x_N,A_T(z_N))$ it is natural to estimate $g_{T}^K$ by the `sample projection'
leading to
\begin{equation*}\label{eq:lat_estimator}
   \hat{g}_{T}^K = \left(\bm{\hat{\alpha}}_K^{\text{\textit{lat}}}\right)^T \mathbf{e}_K^{\text{\textit{lat}}},
\end{equation*}
with
\begin{align*}\label{eq estimator coefficients}
  {\hat{\boldsymbol{\alpha}}}_K^{\text{\textit{lat}}} = \left(\left(\bm{E}_K^{{\text{\textit{lat}}}}\right)^T\bm{E}_K^{{\text{\textit{lat}}}}\right)^{-1} \left(\bm{E}_K^{{\text{\textit{lat}}}}\right)^T\bm{X},
\end{align*}
where $\bm{X} = \left(x_{1}, ..., x_{N}\right)^T$ and $\bm{E}_K^{{\text{\textit{lat}}}}$ is an $N \times K$ matrix with the $n^{\text{th}}$ row equal to $\mathbf{e}_K^{\text{\textit{lat}}}(A_T(z_n))$, $n=1,\ldots,N$.
Notice that ${\hat{\boldsymbol{\alpha}}}_K^{\text{\textit{lat}}}$ corresponds to the usual least squares estimator from a regression of $X$ against $K$ basis functions valued at time $T$.

\section{Convergence Rates for Regress-Later with sieves} \label{sec:convergence_regresslater}
In this section we derive convergence rates for Regress-Later with
sieves and comment on convergence rates for Regress-Now with sieves. We start with the analysis of Regress-Later estimators. Our method of proof follows
\citet{Newey_SeriesEstimators}. Its presentation follows in part
\citet{Hansen}. However, the conditions we impose are different from the
conditions in \citet{Newey_SeriesEstimators} which have, for
instance, also been applied by \citet{Stentoft_LSMC}. To understand this note that \citet{Newey_SeriesEstimators} takes a nonparametric approach to estimating a conditional expectation that is unknown. We exemplify this with Assumption 3 in \citet{Newey_SeriesEstimators}. This assumption with $d=0$ (not to be confused with the $d$ we use here for the dimension of $Z$) would read as follows for Regress-Later\\

{\it \noindent There are $\gamma > 0$, $\bm{\alpha}_K^{\textit{lat}}$ s.t.
\begin{equation}\label{eq Newey sup approx}
    \sup_{x \in D} \abs{g_T(x) - g_T^K(x)} = \sup_{x \in D} \abs{g_T(x) - \left(\boldsymbol{\alpha}_K^{\text{\textit{lat}}}\right)^T\bm{e}_K^{\text{\textit{lat}}}(x)}= O(K^{-\gamma_{\textit{lat}}})
\end{equation}
as $K \rightarrow \infty$, where $D$ is the domain of $g_T$.}\\ 

Note that Condition \eqref{eq Newey sup approx} is independent of the probability measure $\tilde{\PP}$. From a nonparametric point of view this makes perfectly sense,
because, if it is fulfilled, the convergence rate is the same whatever the true probability measure. However, in the context of LSMC {\it we do know}\ $\tilde{\PP}$, because it is
the measure used in the simulation and it is controlled by the user. Thus, it is legitimate to relax Assumption 3 in \citet{Newey_SeriesEstimators}. Additionally, Condition \eqref{eq Newey sup approx} implicitly requires that $g_T$ is bounded or that $D$ is compact, unless $g_T$ is, for
instance, in the span of the $\bm{e}_K^{\text{\textit{lat}}}$. In the context of American option pricing \citet{Stentoft_LSMC} circumvents the problem by explicitly neglecting far in-the-money and far out-of-the-money tails of the distribution. Although this is a reasonable assumption in the context of American options obtaining results on the whole domain is surely welcomed in other areas of application. As we know $\tilde{\PP}$ in a simulation-based framework, we will replace Assumption 3 in \citet{Newey_SeriesEstimators} by the following condition\\

\begin{assumption} \label{as:Newey_3 modified}
There are $\gamma_{\textit{lat}} > 0$, $\bm{\alpha}_K^{\textit{lat}}$ s.t.
\begin{eqnarray*}
  \sqrt{ \EE_{\tilde{\PP}}\left[\left(g_T(A_T(Z)) - (\bm{\alpha}_K^{\text{\textit{lat}}})^T \mathbf{e}_K^{\text{\textit{lat}}}(A_T(Z))\right)^4\right]} & = & \sqrt{\int_{\RR^{\ell}} \left(g_T(u) - (\bm{\alpha}_K^{\text{\textit{lat}}})^T \mathbf{e}_K^{\text{\textit{lat}}}(u)\right)^4 \, \intd \tilde{\PP}^{A_T(Z)}(u)} \nonumber \\
  & = & \sqrt{\int_{\RR^{\ell}} a_T^K(u)^4 \, \intd \tilde{\PP}^{A_T(Z)}(u)} =O\big(K^{- \gamma_{\text{\textit{lat}}}}\big).
\end{eqnarray*}
\end{assumption}

Notice that Assumption \ref{as:Newey_3 modified} does not require that $g_T$ is bounded or that its domain is compact.
From a nonparametric point of view  Assumption \ref{as:Newey_3 modified} is unsatisfactory, because the $O\big(K^{-\gamma_{\text{\textit{lat}}}}\big)$ term on the right-hand side is not independent of $\tilde{\PP}$. However, with regard to LSMC it is worth stressing once again that we know the data generating process, i.e.~$\tilde{\PP}$, $A_T$ and $g_T$, and that therefore Assumption \ref{as:Newey_3 modified} can, for instance, be checked by considering the behavior of $g_T$ in the tails. Moreover, it is also worth pointing out that
$$
\sqrt{\int_{\R^{\ell}} \left(g_T(u) - g_T^K(u)\right)^4 \, \intd \tilde{\PP}^{A_T(Z)}(u)}=O(K^{- \gamma_{\text{\textit{lat}}}})
$$
implies
$$
\sqrt{\int_{\R^{\ell}} \left(g_T(u) - g_T^K(u)\right)^4 \, \intd \widehat{\PP}^{A_T(Z)}(u)}=O(K^{- \gamma_{\text{\textit{lat}}}})
$$
whenever $\widehat{\PP}^{A_T(Z)}$ has a bounded density w.r.t.~$\tilde{\PP}^{A_T(Z)}$.\\
We now state our second assumption to derive the convergence rate for Regress-Later with sieves.

\begin{assumption} \label{as:Newey_1}
$\left((X_{1}, A_T(Z_1)),\ldots,(X_{N}, A_T(Z_{N}))\right)$ are i.i.d.
\end{assumption}

To formulate our theorem on the convergence rate of Regress-Later with sieves we define the net $\tilde{h}^{\text{\textit{lat}}}: \NN \times \NN \rightarrow \RR$ by
$$ \tilde{h}^{\text{\textit{lat}}}(N,K):= \frac{1}{N} \; \EE_{\tilde{\PP}}\left[\left(\left(\bm{e}_K^{\text{\textit{lat}}}(A_T(Z))\right)^T \bm{e}_K^{\text{\textit{lat}}}(A_T(Z))\right)^2\right].$$
We can now state our main theorem for the convergence of Regress-Later estimators.

\begin{theorem} \label{th:Newey_1}
    Let Assumptions \ref{as:Newey_3 modified} and \ref{as:Newey_1} be satisfied. Additionally, assume that there is a sequence $\mathcal{K}: \NN \rightarrow \NN$ such that
\begin{equation}\label{th:Newey eq rate}
\tilde{h}^{\text{\textit{lat}}}(N,\mathcal{K}(N)) \rightarrow 0 \mbox{ as } N \rightarrow \infty.
\end{equation}
Then
    \begin{equation*} \label{eq:Newey_1}
        \EE_{\tilde{\PP}}\left[ \left(X - \hat{g}_T^{\mathcal{K}(N)}(A_T(Z))\right)^2 \right] =
        O_{\tilde{\PP}}\left(\mathcal{K}(N)^{-\gamma_{\text{\textit{lat}}}}\right).
    \end{equation*}
\end{theorem}

Please notice that the convergence rate in Theorem \ref{th:Newey_1} is completely determined by the speed of the approximation and Condition (\ref{th:Newey eq rate}), i.e.~the growth rate of $\mathcal{K}(N)$. Below, we will see that this is in sharp contrast to Regress-Now. Moreover, this fact makes it possible that we may obtain a convergence rate that is faster than $N^{-1}$. We will comment on this at the end of this section.\\
Before giving the proof of Theorem \ref{th:Newey_1} two further remarks are in order. First, notice that Condition (\ref{th:Newey eq rate}) restricts the growth rate of $\mathcal{K}$ when compared to $N$, because for fixed $N$ the net $\tilde{h}^{\text{\textit{lat}}}(N,K)$ is increasing in $K$. Second, if Condition (\ref{th:Newey eq rate}) holds, then we also have

$$\widehat{h}^{\text{\textit{lat}}}(N,\mathcal{K}(N)):= \frac{1}{N} \; \EE_{\widehat{\PP}}\left[\left( \left( \bm{e}_{\mathcal{K}(N)}^{\text{\textit{lat}}}(A_T(Z)) \right)^T \bm{e}_{\mathcal{K}(N)}^{\text{\textit{lat}}}(A_T(Z))\right)^2\right] \rightarrow 0 \mbox{ as } N \rightarrow \infty$$
whenever $\widehat{\PP}^{A_T(Z)}$ has a bounded density w.r.t.~$\tilde{\PP}^{A_T(Z)}$.  \\

The proof of Theorem \ref{th:Newey_1} is based on the following two lemmas whose proofs are given in the appendix.

\begin{lemma}\label{lemma conv matrix and smallest eigenvalue} If Condition (\ref{th:Newey eq rate}) and Assumption \ref{as:Newey_1} hold, we have
\begin{equation}\label{eq matrix conv}
\Big|\Big| \frac{1}{N}  \left(\bm{E}_{\mathcal{K}(N)}^{{\text{\textit{lat}}}}\right)^T\bm{E}_{\mathcal{K}(N)}^{{\text{\textit{lat}}}}-I_{\mathcal{K}(N)}\Big|\Big|_F=o_{\tilde{\PP}}(1),
\end{equation}
where $||\cdot||_F$ is the Frobenius norm and $I_{\mathcal{K}(N)}$ denotes the $\mathcal{K}(N) \times \mathcal{K}(N)$ identity matrix. Moreover,
\begin{equation}\label{eq eigenvalue conv}
\lambda_{\min}\left(\frac{1}{N} \left(\bm{E}_{\mathcal{K}(N)}^{{\text{\textit{lat}}}}\right)^T\bm{E}_{\mathcal{K}(N)}^{{\text{\textit{lat}}}}\right) \stackrel{\tilde{\PP}}{\rightarrow} 1,
\end{equation}
where $\lambda_{\min}(\bm{A})$ denotes the smallest eigenvalue of a matrix $\bm{A}$.
\end{lemma}

\begin{lemma}\label{lemma conv quadratic form parameter} If Condition (\ref{th:Newey eq rate}) and Assumptions \ref{as:Newey_3 modified} and  \ref{as:Newey_1} hold, we have
\begin{equation*}
\left( {\hat{\boldsymbol{\alpha}}}_{\mathcal{K}(N)}^{\text{\textit{lat}}}- \boldsymbol{\alpha}_{\mathcal{K}(N)}^{\text{\textit{lat}}} \right)^T\left( {\hat{\boldsymbol{\alpha}}}_{\mathcal{K}(N)}^{\text{\textit{lat}}}- \boldsymbol{\alpha}_{\mathcal{K}(N)}^{\text{\textit{lat}}} \right)=o_{\tilde{\PP}}\big({\mathcal{K}(N)}^{- \gamma_{\text{\textit{lat}}}}\big).
\end{equation*}
\end{lemma}

\noindent Lemma \ref{lemma conv matrix and smallest eigenvalue} shows that the sample second moment matrix of the basis functions converges to the identity matrix in the Frobenius norm. Lemma \ref{lemma conv quadratic form parameter} considers the convergence of the estimation error of the estimated coefficients. We now give the proof of the above theorem.

\begin{proof}[Proof of Theorem ~\ref{th:Newey_1}] Observe first that Assumption \ref{as:Newey_3 modified} implies
\begin{equation}\label{first equation main thm}
\EE_{\tilde{\PP}}\left[\left(g_T(A_T(Z)) - \left(\bm{\alpha}_K^{\text{\textit{lat}}}\right)^T \mathbf{e}_K^{\text{\textit{lat}}}(A_T(Z))\right)^2\right] \leq O\big(K^{- \gamma_{\text{\textit{lat}}}}\big)
\end{equation}
by Cauchy-Schwarz and that
$$
\EE_{\tilde{\PP}}\left[e_k^{\text{\textit{lat}}}(A_T(Z))\left(g_T(A_T(Z)) - (\bm{\alpha}_K^{\text{\textit{lat}}})^T \mathbf{e}_K^{\text{\textit{lat}}}(A_T(Z))\right)\right]=0,\, k=1,\ldots,K.
$$
Hence,
\begin{eqnarray}\label{eq tm.newey proof}
        \EE_{\tilde{\PP}}\left[ \left(g_T(A_T(Z))-\hat{g}_T^{\mathcal{K}(N)}(A_T(Z))\right)^2 \right]
        &  = &\EE_{\tilde{\PP}}\left[\left(g_T(A_T(Z)) - \left(\bm{\alpha}_{\mathcal{K}(N)}^{\text{\textit{lat}}}\right)^T \mathbf{e}_{\mathcal{K}(N)}^{\text{\textit{lat}}}(A_T(Z))\right)^2\right] \nonumber \\
&& {} +\EE_{\tilde{\PP}}\left[\left(\left(\bm{\alpha}_{\mathcal{K}(N)}^{\text{\textit{lat}}}\right)^T \mathbf{e}_{\mathcal{K}(N)}^{\text{\textit{lat}}}(A_T(Z))- \hat{g}_T^{\mathcal{K}(N)}(A_T(Z))\right)^2\right]  \nonumber \\
        & \leq & O\big(\mathcal{K}(N)^{- \gamma_{\text{\textit{lat}}}}\big) + \left( {\hat{\boldsymbol{\alpha}}}_{\mathcal{K}(N)}^{\text{\textit{lat}}}- \boldsymbol{\alpha}_{\mathcal{K}(N)}^{\text{\textit{lat}}} \right)^T\left( {\hat{\boldsymbol{\alpha}}}_{\mathcal{K}(N)}^{\text{\textit{lat}}}- \boldsymbol{\alpha}_{\mathcal{K}(N)}^{\text{\textit{lat}}} \right) \nonumber \\
        & = & O\big(\mathcal{K}(N)^{- \gamma_{\text{\textit{lat}}}}\big) + o_{\tilde{\PP}}\big(\mathcal{K}(N)^{- \gamma_{\text{\textit{lat}}}}\big)\nonumber \\
        & = & O_{\tilde{\PP}}(\mathcal{K}(N)^{- \gamma_{\text{\textit{lat}}}}),
\end{eqnarray}
where the inequality follows from (\ref{first equation main thm}) and the second equality from Lemma \ref{lemma conv quadratic form parameter}.
\end{proof}

The first equality in (\ref{eq tm.newey proof}) nicely illustrates that the sieve estimator is subject to two errors: an \textit{approximation error} 
$$\EE_{\tilde{\PP}}\left[\left(g_T(A_T(Z)) - \left(\bm{\alpha}_{\mathcal{K}(N)}^{\text{\textit{lat}}}\right)^T \mathbf{e}_{\mathcal{K}(N)}^{\text{\textit{lat}}}(A_T(Z))\right)^2\right],$$ 
and an \textit{estimation error} 
$$\EE_{\tilde{\PP}}\left[ \left(\left(\bm{\alpha}_{\mathcal{K}(N)}^{\text{\textit{lat}}}\right)^T \mathbf{e}_{\mathcal{K}(N)}^{\text{\textit{lat}}}(A_T(Z)) - \hat{g}_T^{\mathcal{K}(N)}(A_T(Z)) \right)^2\right].$$
It is worth emphasizing once more that for Regress-Later both are entirely driven by the speed of the approximation error and the growth rate of $\mathcal{K}(N)$ only. The fact that the estimation error is entirely driven by the speed of the approximation error and the growth rate of $\mathcal{K}(N)$ is a result of the fact that Equation (\ref{eq:regresseq-lat}) describes a nonstandard regression problem. Indeed, as $N$ increases the variance of the noise term, i.e.~$a_T(A_T(Z))$, converges to zero. We further comment on this at the end of this section.\\

Let us now discuss convergence rates for Regress-Now with sieves. We argued above that Regress-Later with sieves need not be considered as a nonparametric problem, because we do know $g_T$, $A_T$ and $\tilde{\PP}$. We therefore argued that we may replace the assumptions typically imposed in nonparametric settings by weaker ones. The situation is (slightly) different for Regress-Now with sieves. There we are interested in $g_{0,t}$ which is not given. Although, $g_{0,t}$ depends only on $g_T$, $A_T(Z)$, $\tilde{\PP}$ and the information generated by $Z$ up to time $t$, the problem of assessing certain properties to $g_{0,t}$ might be rather complicated so that one would tend to consider the problem as nonparametric even so $\tilde{\PP}$ is known. In such a case, convergence rates and conditions needed to obtain these rates may be taken from Theorem 1 in \citet{Newey_SeriesEstimators}. However, if we have some knowledge about $g_{0,t}$ (which is not unlikely since the problem is not nonparametric), then we may weaken the assumptions imposed in \citet{Newey_SeriesEstimators} similar as we did above for Regress-Later with sieves.

We now briefly outline the Regress-Now approach under assumptions
 similar to Assumption \ref{as:Newey_3 modified} and Condition (\ref{th:Newey eq rate}).
 \begin{assumption} \label{as:Newey_3 modified regress-now}
There are $\gamma_{\text{\textit{now}}} > 0$,
$\bm{\alpha}_K^{\text{\textit{now}}}$ s.t.
\begin{eqnarray*}
  \sqrt{ \EE_{\tilde{\PP}}\left[\left(g_{0,t}(A_t(Z)) - (\bm{\alpha}_K^{\text{\textit{now}}})^T \mathbf{e}_K^{\text{\textit{now}}}(A_t(Z))\right)^4\right]}
   & = & \sqrt{\int_{\RR^{s}} \left(g_{0,t}(u) - (\bm{\alpha}_K^{\text{\textit{now}}})^T \mathbf{e}_K^{\text{\textit{now}}}(u)\right)^4 \, \intd \tilde{\PP}^{A_t(Z)}(u)} \nonumber \\
  & = & \sqrt{\int_{\RR^{s}} a_{0,t}^K(u)^4 \, \intd \tilde{\PP}^{A_t(Z)}(u)} =O\big(K^{- \gamma_{\text{\textit{now}}}}\big).
\end{eqnarray*}
\end{assumption}

\noindent We also assume that
\begin{assumption} \label{as:Newey_1 regress-now}
$\left((X_{1}, A_t(Z_1)),\ldots,(X_{N}, A_t(Z_{N}))\right)$ are
i.i.d.~and $\EE_{\tilde{\PP}}\left[\big(p_{0,t}(A_T (Z))\big)^2|A_t(Z)\right]=\sigma^2$.
\end{assumption}

\noindent Introduce the net $\tilde{h}^{\text{\textit{now}}}: \NN \times
\NN \rightarrow \RR$ by
$$ \tilde{h}^{\text{\textit{now}}}(N,K):= \frac{1}{N} \; \EE_{\tilde{\PP}}\left[\left( \left( \bm{e}_K^{\text{\textit{now}}}(A_t(Z)) \right)^T \bm{e}_K^{\text{\textit{now}}}(A_t(Z))\right)^2\right].$$
We can now state

\begin{theorem}\label{th:Newey_1 regress now}
    Let Assumptions \ref{as:Newey_3 modified regress-now} and \ref{as:Newey_1 regress-now} be satisfied. Additionally, assume that there is a sequence $\mathcal{K}: \mathds{N} \rightarrow \mathds{N}$ such that
\begin{equation}\label{th:Newey eq rate regress now}
\tilde{h}^{\text{\textit{now}}}(N,\mathcal{K}(N)) \rightarrow 0
\mbox{ as } N \rightarrow \infty.
\end{equation}
Then
    \begin{equation} \label{eq:Newey_1 regress now}
        \EE_{\tilde{\PP}}\left[ \left(g_{0,t}(A_t(Z)) - \hat{g}_{0,t}^{\mathcal{K}(N)}(A_t(Z))\right)^2 \right] =
        O_{\tilde{\PP}}\left(\frac{\mathcal{K}(N)}{N}+\mathcal{K}(N)^{-\gamma_{\text{\textit{now}}}}\right).
    \end{equation}
\end{theorem}
The result corresponds to Theorem 1 in \citet{Newey_SeriesEstimators}, but requires weaker assumptions; see also \citet{Stentoft_LSMC} for a similar result, where a nonparametric setting was used for the pricing of American options.
The result differs from the result of \citet{Newey_SeriesEstimators} and \citet{Stentoft_LSMC} as the convergence speed, $\mathcal{K}(N)\slash N+\mathcal{K}(N)^{-\gamma_{\text{\textit{now}}}}$, in \eqref{eq:Newey_1 regress now} is not independent of $\tilde{\PP}$. Notice the appearance of the term $\mathcal{K}(N) \slash N$ in the convergence rate of the Regress-Now estimator. This term does not appear in the convergence rate of the Regress-Later estimator. We further explain this difference at the end of this section.

The proof of Theorem \ref{th:Newey_1 regress now} is based on the following two lemmas. The first lemma is very similar to Lemma \ref{lemma conv matrix and smallest eigenvalue}. Its proof follows
along the lines of the proof of Lemma \ref{lemma conv matrix and smallest eigenvalue} and is therefore omitted. The second lemma is different from its counterpart for Regress-Later and it explains why we obtain
the term $\mathcal{K}(N) \slash N$ in Equation \eqref{eq:Newey_1 regress now}. Its proof is given in the appendix.

\begin{lemma}\label{lemma conv matrix and smallest eigenvalue regress now} If Condition (\ref{th:Newey eq rate regress now}) and Assumption \ref{as:Newey_1 regress-now} hold, we have
\begin{equation*}\label{eq matrix conv now}
\Big|\Big| \frac{1}{N}
\left(\bm{E}_{\mathcal{K}(N)}^{{\text{\textit{now}}}}\right)^T\bm{E}_{\mathcal{K}(N)}^{{\text{\textit{now}}}}-I_{\mathcal{K}(N)}\Big|\Big|_F=o_{\tilde{\PP}}(1),
\end{equation*}
where $||\cdot||_F$ is again the Frobenius norm. Moreover,
\begin{equation*}\label{eq eigenvalue conv regress now}
\lambda_{\min}\left( \frac{1}{N} \left(\bm{E}_{\mathcal{K}(N)}^{{\text{\textit{now}}}}\right)^T\bm{E}_{\mathcal{K}(N)}^{{\text{\textit{now}}}}\right)
\stackrel{\tilde{\PP}}{\rightarrow} 1,
\end{equation*}
where $\lambda_{\min}(\bm{A})$ denotes again the smallest eigenvalue
of a matrix $\bm{A}$.
\end{lemma}

\begin{lemma}\label{lemma conv quadratic form parameter regress now} If Condition (\ref{th:Newey eq rate regress now}) and Assumptions \ref{as:Newey_3 modified regress-now} and  \ref{as:Newey_1 regress-now}
hold, we have
\begin{equation*}
\left( {\hat{\boldsymbol{\alpha}}}_{\mathcal{K}(N)}^{\text{\textit{now}}}-
\boldsymbol{\alpha}_{\mathcal{K}(N)}^{\text{\textit{now}}} \right)^T\left(
{\hat{\boldsymbol{\alpha}}}_{\mathcal{K}(N)}^{\text{\textit{now}}}-
\boldsymbol{\alpha}_{\mathcal{K}(N)}^{\text{\textit{now}}}
\right)=O_{\tilde{\PP}}\left(\frac{\mathcal{K}(N)}{N}\right)+o_{\tilde{\PP}}\big({\mathcal{K}(N)}^{-
\gamma_{\text{\textit{now}}}}\big).
\end{equation*}
\end{lemma}

\noindent We now give the proof of the above theorem.

\begin{proof}[Proof of Theorem ~\ref{th:Newey_1 regress now}] Similar to the proof of Theorem \ref{th:Newey_1} observe first that Assumption \ref{as:Newey_3 modified regress-now} implies
\begin{equation}\label{first equation main thm regress now}
\EE_{\tilde{\PP}}\left[\left(g_{0,t}(A_t(Z)) -
\left(\bm{\alpha}_K^{\text{\textit{now}}}\right)^T
\mathbf{e}_K^{\text{\textit{now}}}(A_t(Z))\right)^2\right] \leq
O\big(K^{- \gamma_{\text{\textit{now}}}}\big)
\end{equation}
by Cauchy-Schwarz. Moreover,
$$
\EE_{\tilde{\PP}}\left[e_k^{\text{\textit{now}}}(A_t(Z))\left(g_{0,t}(A_t(Z))
- \left(\bm{\alpha}_K^{\text{\textit{now}}}\right)^T
\mathbf{e}_K^{\text{\textit{now}}}(A_t(Z))\right)\right]=0,\,
k=1,\ldots,K.
$$
Hence,
\begin{eqnarray}\label{eq tm.newey proof regress now}
       \EE_{\tilde{\PP}}\left[ \left(g_{0,t}(A_t(Z))-\hat{g}_{0,t}^{\mathcal{K}(N)}(A_t(Z))\right)^2 \right]
        &  = & \EE_{\tilde{\PP}}\left[\left(\left(\bm{\alpha}_{\mathcal{K}(N)}^{\text{\textit{now}}}\right)^T \mathbf{e}_{\mathcal{K}(N)}^{\text{\textit{now}}}(A_t(Z))- \hat{g}_{0,t}^{\mathcal{K}(N)}(A_t(Z))\right)^2\right] \nonumber \\
&& {} +\EE_{\tilde{\PP}}\left[\left(g_{0,t}(A_t(Z)) - (\bm{\alpha}_{\mathcal{K}(N)}^{\text{\textit{now}}})^T \mathbf{e}_{\mathcal{K}(N)}^{\text{\textit{lat}}}(A_t(Z))\right)^2\right] \nonumber \\
        & \leq & \left( {\hat{\boldsymbol{\alpha}}}_{\mathcal{K}(n)}^{\text{\textit{now}}}- \boldsymbol{\alpha}_{\mathcal{K}(n)}^{\text{\textit{now}}} \right)^T\left( {\hat{\boldsymbol{\alpha}}}_{\mathcal{K}(N)}^{\text{\textit{now}}}- \boldsymbol{\alpha}_{\mathcal{K}(N)}^{\text{\textit{now}}} \right) \nonumber \\
&& {}+ O\big(\mathcal{K}(N)^{- \gamma_{\text{\textit{now}}}}\big) \nonumber \\
        & = & O_{\tilde{\PP}}\left(\frac{\mathcal{K}(N)}{N}\right) + O\big(\mathcal{K}(N)^{- \gamma_{\text{\textit{now}}}}\big),
\end{eqnarray}
where the inequality follows from (\ref{first equation main thm regress now})
and the second equality from Lemma \ref{lemma conv quadratic form parameter regress now}.
\end{proof}
As for Regress-Later the first equality in (\ref{eq tm.newey proof regress now}) illustrates that also Regress-Now is subject to two errors: an \textit{approximation error} $$\EE_{\tilde{\PP}}\left[\left(g_{0,t}(A_t(Z)) - (\bm{\alpha}_{\mathcal{K}(N)}^{\text{\textit{now}}})^T \mathbf{e}_{\mathcal{K}(N)}^{\text{\textit{lat}}}(A_t(Z))\right)^2\right],$$ 
and an \textit{estimation error} 
$$\EE_{\tilde{\PP}}\left[\left(\left(\bm{\alpha}_{\mathcal{K}(N)}^{\text{\textit{now}}}\right)^T \mathbf{e}_{\mathcal{K}(N)}^{\text{\textit{now}}}(A_t(Z))- \hat{g}_{0,t}^{\mathcal{K}(N)}(A_t(Z))\right)^2\right].$$
Notice that as for Regress-Later the approximation error is also driven by the speed of the approximation error and the growth rate of $\mathcal{K}(N)$. However, note also the difference compared to Regress-Later: The estimation error is driven by the ratio of $\mathcal{K}(N)$ to $N$ and the approximation error. The difference can also be seen from the following equations where we omitted the superscripts $\textit{now}$ and $\textit{lat}$
\begin{eqnarray*}
 \left( {\hat{\boldsymbol{\alpha}}}_K-\boldsymbol{\alpha}_K \right)&=& (\bm{E}_K^T \bm{E}_K)^{-1}\bm{E}_K^T(\bm{X}-\bm{E}_K\boldsymbol{\alpha}_K) \\
                                                                &=& (\bm{E}_K^T \bm{E}_K)^{-1}\bm{E}_K^T ((\bm{X}-\bm{E}\boldsymbol{\alpha})+(\bm{E}\boldsymbol{\alpha}-\bm{E}_K\boldsymbol{\alpha}_K)) \\
                                                                &=& (\bm{E}_K^T \bm{E}_K)^{-1}\bm{E}_K^T (\bm{p}+\bm{a}^K),
\end{eqnarray*}
where $\bm{E}$ is an infinite-dimensional matrix containing all basis functions and $\bm{\alpha}$ is the true (infinite-dimensional) parameter vector. Here $\bm{p}$ gives the projection error which is zero for Regress-Later. However, it is unequal to zero for Regress-Now. Moreover, for both Regress-Now and Regress-Later the variance of the approximation error $a^{K}$ converges to zero. However, as can be seen from the proof of Lemma \ref{lemma conv quadratic form parameter regress now} in the Appendix \ref{proofs}, it is the projection error that contributes the rate $\mathcal{K}(N) \slash N$ to  the estimation error in Regress-Now.\\
The absence of the term $\mathcal{K}(N) \slash N$ in the mean-square error of Regress-Later makes it plausible that the Regress-Later estimator may potentially converge faster than the Regress-Now estimator. We deliberately state here ``potentially'' as the ultimate convergence rate depends on the $\gamma_{\textit{now}}$ and $\gamma_{\textit{lat}}$ which are problem-dependent. In particular, the choice of basis plays an important role in the determination of $\gamma_{\textit{now}}$ and $\gamma_{\textit{lat}}$. However, it is clear that the Regress-Now convergence rate can never be faster than $N^{-1}$. This follows simply from the fact that the best we can hope for is that $g_{0,t}$ is contained in the span of finitely many basis functions. Then the approximation error vanishes and we are left with the rate $N^{-1}$. In contrast, in Regress-Later if Condition (\ref{th:Newey eq rate}) is fulfilled with $\mathcal{K}(N) \propto N^a$ for some $0<a<1$, then the convergence rate for Regress-Later equals $N^{-a \; \gamma_{\textit{lat}}}$. We can see that for the right combination of $a$ and $\gamma_{\textit{lat}}$ it is possible to achieve a convergence rate that is even faster than $N^{-1}$. An example will be provided in Section \ref{sec:piecewise linear}.\\

We finally comment on the fact that the discussed convergence rates pertain to slightly different problems. The speed of convergence for the Regress-Now estimator refers to convergence to the conditional expectation function $g_{0,t}$. On the contrary, the discussed convergence rate for the Regress-Later estimator pertains to convergence to the payoff function $X$. As discussed in Section \ref{sec:regress-later} in Regress-Later we achieve the approximation to the conditional expectation function by applying the conditional expectation operator to the estimated payoff function, $\hat{g}_T^{\mathcal{K}(N)}$. We thereby do not incur a projection error as long as the conditional expectations of the basis functions have closed-form solutions. We can show that the convergence rate of the conditional expectation of the Regress-Later estimator to the conditional expectation of $X$, i.e. $g_{0,t}(A_t(Z))$ is implied by the convergence of the Regress-Later estimator to $X$. More explicitly we have

\begin{eqnarray*}
    \EE_{\tilde{\PP}} \left[ \left( g_{0,t}(A_t(Z)) - \EE_{\tilde{\PP}}\left[ \hat{g}_T^{\mathcal{K}(N)}(A_T(Z)) \big | \eFF_t\right]\right)^2\right]  
    &=& \EE_{\tilde{\PP}}\left[ \left( \EE_{\tilde{\PP}}\left[ X-\hat{g}_T^{\mathcal{K}(N)}(A_T(Z)) \big | \eFF_t\right]\right)^2\right] \\
    &\leq& \EE_{\tilde{\PP}} \left[ \EE_{\tilde{\PP}} \left[ \left( X- \hat{g}_T^{\mathcal{K}(N)}(A_T(Z))\right)^2\big | \eFF_t \right] \right] \\
    &=& \EE_{\tilde{\PP}} \left[ \left( X- \hat{g}_T^{\mathcal{K}(N)}(A_T(Z))\right)^2\right],
\end{eqnarray*}
where the first inequality follows from Jensen's inequality for conditional expectations and the last equality uses the projection law of expectations.


\section{Orthonormal piecewise linear functions as sieves}\label{sec:piecewise linear}
In this section, we show that the convergence rate in mean-square for Regress-Later can indeed be faster than $N^{-1}$. To present this claim in a very simple set-up and to avoid technicalities that are of no relevance for our claim we consider a compact interval.
The convergence rate highly depends on the properties of the basis. Typically applied bases are polynomials.
Here we consider a basis consisting of piecewise linear functions, for which both the construction as well as the analysis in view of establishing convergence rates simplifies. Moreover, for this basis it is reasonable to expect that we can compute conditional expectations exactly. Importantly, the suggested basis is by construction orthogonal and can be easily set up. Moreover, the results of Theorem \ref{th:Newey_1} can be explicitly calculated for the piecewise linear functions as the derived $\gamma_{\textit{lat}}$ applies to a large class of functions.

We now outline the Regress-Later estimation with piecewise linear functions as sieves.
Let $D=[a_1,a_2] \subset \RR$ denote the support of $g_T(A_T(Z))$. We construct an orthonormal basis on $L_2(D, \mathcal{B}(D), \tilde{\PP}^{A_T(Z)})$ based on non-overlapping linear functions.
 We require the following assumption in order to construct our basis functions

\begin{assumption}\label{as:pos dens pw}
     $g_T(A_T(Z))$ has a density w.r.t.~Lebesgue measure which is a positive continuous function on $D$.
\end{assumption}
Then, the domain $D$ is chopped into $K$ intervals, $[b_k,b_{k+1})$, $k=1,\ldots,K+1$, such that $\text{\textit{Pr}}\left(b_k \leq A_T(Z) < b_{k+1}\right)= 1 \slash K$, $\forall k=1,...,K$. Assumption \ref{as:pos dens pw} ensures that as the truncation parameter $K$ grows, the intervals can be made arbitrarily small and cover each probability $1 \slash K$.  Define $K$ non-overlapping indicator functions

\begin{equation*}
    \Ind_{k}^{\textit{lat}}(u) := 
    \begin{cases} 1 & \text {if } u \in [b_{k}, b_{k+1}) \\
    0 & \text{otherwise}
    \end{cases}
\end{equation*}
for $k=1,..., K$.
By construction the indicator functions are orthogonal. On each interval two basis functions are now defined:
\begin{align*}
    e_{0k}^{\textit{lat}}(u):=C_{0k} \Ind_k^{\textit{lat}}(u) \\
    e_{1k}^{\textit{lat}}(u) :=C_{1k} \Ind_k^{\textit{lat}}(u) (u-c_k),
\end{align*}
where $C_{0k}$, $C_{1k}$ and $c_k$ are chosen such that $e_{0k}(A_T(Z))$ and $e_{1j}(A_T(Z))$ are orthonormal $\forall k, j$. Hence, $C_{0k}=\sqrt{K}$, $C_{1k}=1 \slash \sqrt{\EE_{\tilde{\PP}}\left[\Ind_k^{\textit{lat}}(A_T(Z))(A_T(Z)-c_k)^2\right]}$ and $c_k = K \, \EE_{\tilde{\PP}}[\Ind_k^{\textit{lat}}(A_T(Z)) A_T(Z)]$.
By construction we then have the following orthonormality results
\begin{align*}
    & \EE_{\tilde{\PP}}\left[e_{0k}^{\textit{lat}}(A_T(Z)) e_{0j}^{\textit{lat}}(A_T(Z))\right]= \delta_{kj} \\
    & \EE_{\tilde{\PP}}\left[e_{1k}^{\textit{lat}}(A_T(Z)) e_{1j}^{\textit{lat}}(A_T(Z))\right]= \delta_{kj} \\
    & \EE_{\tilde{\PP}}\left[e_{0k}^{\textit{lat}}(A_T(Z)) e_{1j}^{\textit{lat}}(A_T(Z))\right] = 0,
\end{align*}
where $\delta_{kj}$ denotes the Kronecker delta.

\begin{assumption}\label{as:bounded_2ndderivative}
$g_T$ is twice continuously differentiable on $(a_1,a_2)$ and there is a $B < \infty$ such that $\sup_{u \in (a_1,a_2)} \abs{g_T''(u)} \leq B$.
\end{assumption}

\begin{lemma}\label{lem:approx_error_piecewise}
If Assumptions \ref{as:pos dens pw} and \ref{as:bounded_2ndderivative} hold, the deterministic approximation error vanishes as $K \to \infty$;
\begin{equation*} \label{eq:truncation_convergence_rate_d}
    \sqrt{\EE_{\tilde{\PP}} \left[ \left( g_T(A_T(Z)) - g_T^K(A_T(Z)) \right)^4 \right]} = O(K^{-4}).
\end{equation*}
\end{lemma}

\begin{proof}
See Appendix \ref{proofs}.
\end{proof}
\noindent Consequently, Assumption \ref{as:Newey_3 modified} is satisfied with $\gamma_{\textit{lat}}=4$. For Equation \eqref{th:Newey eq rate} we obtain
\begin{eqnarray}\label{eq:tildeh pw}
    \tilde{h}^{\textit{lat}}(N,\mathcal{K}(N))
&= & \frac{1}{N} \EE_{\tilde{\PP}} \left[ \left( \sum_{k=1}^{\mathcal{K}(N)} \left( e_{0k}^{\textit{lat}} (A_T(Z))^2 + e_{1k}^{\textit{lat}}(A_T(Z))^2 \right) \right)^2 \right] \nonumber \\
    &= & \frac{1}{N} \EE_{\tilde{\PP}} \left[ \sum_{\ell=1}^{\mathcal{K}(n)} \sum_{j=1}^{\mathcal{K}(n)} \left( e_{0\ell}^{\textit{lat}}(A_T(Z))^2 + e_{1\ell}^{\textit{lat}}(A_T(Z))^2 \right)\left( e_{0j}^{\textit{lat}}(A_T(Z))^2 + e_{1j}^{\textit{lat}}(A_T(Z))^2 \right) \right] \nonumber \\
    &= &\frac{1}{N} \EE_{\tilde{\PP}} \left[ \sum_{k=1}^{\mathcal{K}(N)} \left( e_{0k}^{\textit{lat}} (A_T(Z))^2 + e_{1k}^{\textit{lat}} (A_T(Z))^2\right)^2\right] \nonumber \\
    &= &\frac{1}{N} \sum_{k=1}^{\mathcal{K}(N)} \EE_{\tilde{\PP}} \Big[ \mathcal{K}(N)^2 \Ind_k^{\textit{lat}}(A_T(Z)) + 2 \mathcal{K}(N) \Ind_k^{\textit{lat}}(A_T(Z)) C_{1k}^2 \left( A_T(Z)-c_k \right)^2 \nonumber \\
    && {}+ C_{1k}^4 \Ind_k^{\textit{lat}}(A_T(Z)) (A_T(Z)-c_k)^4\Big] \nonumber \\
    &= &\frac{1}{N} \sum_{k=1}^{\mathcal{K}(N)}  \Big( \mathcal{K}(N) + 2 \mathcal{K}(N)+ C_{1k}^4  \EE_{\tilde{\PP}}\left[\Ind_k^{\textit{lat}}(A_T(Z)) (A_T(Z)-c_k)^4\right]\Big) \nonumber \\
    & \leq &\frac{3 \; \mathcal{K}(N)^2}{N} + \frac{\mathcal{K}(N)}{N} \max_k \left(\frac{\EE_{\tilde{\PP}}\left[\Ind_k^{\textit{lat}}(A_T(Z)) (A_T(Z)-c_k)^4\right]}{\left( \EE_{\tilde{\PP}}\left[\Ind_k^{\textit{lat}}(A_T(Z)) (A_T(Z)-c_k)^2\right] \right)^2 } \right).
\end{eqnarray}
Assumption \ref{as:pos dens pw} ensures that there is enough variation on each arbitrary interval such that the denominator in \eqref{eq:tildeh pw} is greater than zero. Moreover, it makes sure that the last term in the last line of Equation \eqref{eq:tildeh pw} does not grow faster than the first term in the last line of Equation \eqref{eq:tildeh pw}. Moreover,  the particular growth rate can be determined.

\begin{lemma}\label{lem:ratio moments pw}
If Assumption \ref{as:pos dens pw} is satisfied the following result holds
    \begin{equation*}\label{eq:ratio moments pw}
    \max_{1 \leq k \leq \mathcal{K}(N)}\frac{\EE_{\tilde{\PP}}\left[\Ind_k^{\textit{lat}}(A_T(Z)) (A_T(Z)-c_k)^4\right]}{\left( \EE_{\tilde{\PP}}\left[\Ind_k^{\textit{lat}}(A_T(Z)) (A_T(Z)-c_k)^2\right] \right)^2 } \leq O\left(\mathcal{K}(N)\right).
    \end{equation*}
\end{lemma}

\begin{proof}
See Appendix \ref{proofs}.
\end{proof}

\noindent Hence, by combining Lemma \ref{lem:ratio moments pw} and Equation (\ref{eq:tildeh pw})
\begin{align*}
    \tilde{h}^{\textit{lat}}(N,\mathcal{K}(N)) \leq O\left(\frac{\mathcal{K}(N)^2}{N} \right).
\end{align*}
A sufficient condition for $\tilde{h}^{\textit{lat}}(N,\mathcal{K}(N)) \rightarrow 0$ as $N \rightarrow \infty$ is that $\mathcal{K}(N) \propto N^a$, and with $a < 1 \slash 2$ Condition \eqref{th:Newey eq rate} in Theorem \ref{th:Newey_1} holds. Now Theorem \ref{th:Newey_1} is applicable and gives the convergence rate in mean-square
\begin{equation}\label{eq:conv pw result}
\EE_{\tilde{\PP}}\left[ \left(g_T(A_T(Z)) - \hat{g}_T^{\mathcal{K}(N)}(A_T(Z))\right)^2 \right] = O_{\tilde{\PP}}(\mathcal{K}(N)^{-4}).
\end{equation}
We immediately see that for $\mathcal{K}(N) \propto N^a$ and choosing $a$ only slightly smaller than $1 \slash 2$ we almost achieve a convergence rate of $N^{-2}$, which is considerably faster than the conventional Monte Carlo rate of $N^{-1}$. \\

We now look at the Regress-Later estimator with orthonormal piecewise linear functions as sieve for a specific underlying random variable.\\



\begin{figure}
\begin{center}
    \includegraphics[clip=true, width=0.6\textwidth ]{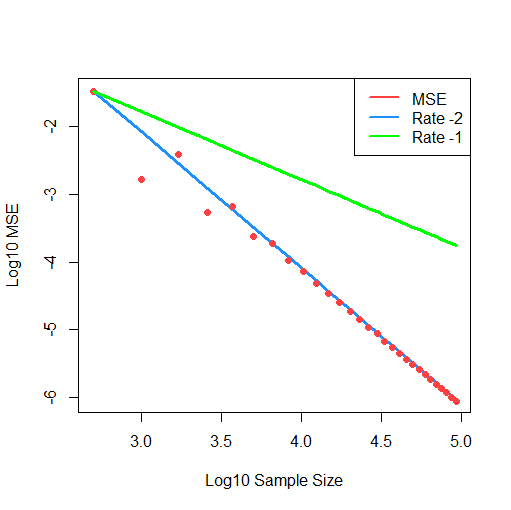} 
    \caption{Regress-Later convergence plot with $K$ up to $30$.}\label{fig:MSE tanh}
\end{center}
\end{figure}

\begin{figure}
\begin{center}
    \includegraphics[clip=true, width=0.6\textwidth ]{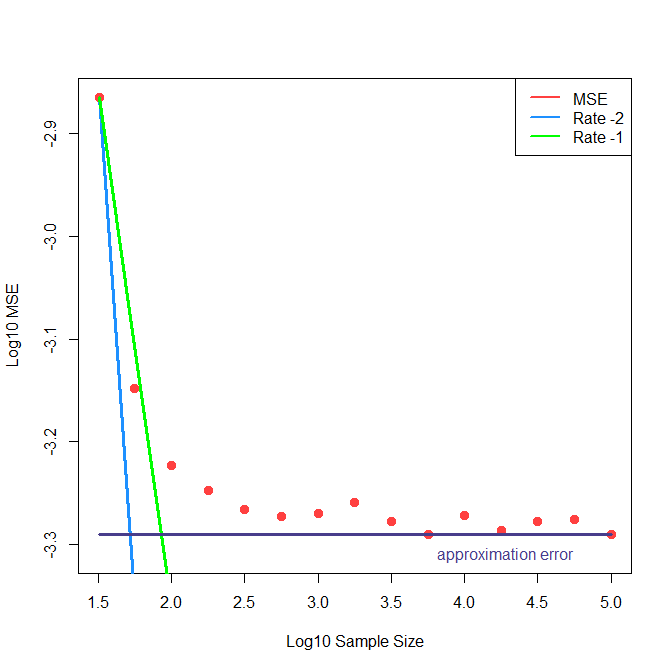} 
    \caption{Regress-Later convergence plot with $K=5$ fixed.}\label{fig:MSE tanh K fixed}
\end{center}
\end{figure}

\noindent{\it Brownian Motion} \\
We consider a Brownian Motion $W(T)$ as the underlying for a function $g_T(W(T))$ fulfilling the necessary conditions for Equation \eqref{eq:conv pw result}. We consider a compact domain $[a_1,a_2]$ for $W(T)$ and define the $K$ intervals $[b_k,b_{k+1}), k=1,\ldots,K$ such that each covers probability $1 \slash K$. The conditional normal density is then $\varphi(w|D) = \varphi(z) \slash \PP(W(T) \in D)$ with $\varphi(w) = \exp(-w^2 \slash (2T)) \slash \sqrt{2\pi\,T}$ the normal density. Assumption \ref{as:pos dens pw} is immediately satisfied. Thus, Lemma \ref{lem:ratio moments pw} applies. Then, again choosing $\mathcal{K}(N) \propto N^{a}$ with $a$ only slightly smaller than $1 \slash 2$ produces a mean square error that converges in probability almost at rate $N^{-2}$.

Figure \ref{fig:MSE tanh} gives the convergence for $X = \tanh(W(10))$ for $K$ up to $30$ and $N = 100 \; K^{2.01}$. We see that we can already achieve the fast convergence rate in finite samples. Figure \ref{fig:MSE tanh K fixed} gives the mean square error for the payoff function where the number of basis functions is fixed at $K=5$ and only the sample size grows up to $10^5$. The example illustrates that the mean square error does not converge further if only the sample size is increased. While the sampling error decreases with the growth of the sample size the approximation error only converges when the number of basis functions grows.

\section{Conclusion} \label{sec:Conclusion}
In this paper the discussion on Regress-Later estimators is picked up and addressed in comparison to Regress-Now estimators, which are currently more popular. Both estimators refer to LSMC solutions. Clarification is given on the functionality of Regress-Now and Regress-Later estimators based on several examples. Examples have been discussed that help to better understand the differences of Regress-Now and Regress-Later. The estimation approach for both estimators is outlined and the regression error for each is specified. It is shown that in Regress-Later the involved regression is nonstandard as the regression error corresponds to the approximation error, which vanishes in the limit. In contrast, the regression error of Regress-Now estimators contains an approximation and a projection error. While the approximation error vanishes in the limit the projection error is not eliminated. This leads to different convergence rates for Regress-Now and Regress-Later estimators. The current literature addresses convergence rates for Regress-Now estimators in a nonparametric setting. In this paper it is shown that the problem specification for Regress-Later is not nonparametric. This allows to relax the conditions typically necessary in nonparametric problems solved with sieve. Moreover, it is indicated that a nonparametric problem specification may also apply to Regress-Now estimators, which then similarly allows for weaker conditions. A specific basis is constructed based on piecewise linear functions and the Regress-Later convergence rate with this basis is derived explicitly. The result shows that Regress-Later estimators can be constructed such that they converge faster than the more often applied Regress-Now estimators.

\section*{Acknowledgements}
The authors are grateful for valuable input from participants at the $7$th World Congress of the Bachelier Finance Society, the Netspar Pension Day $2012$ and the AFMATH conference $2013$ in Brussels, in particular Dilip Madan, Ragnar Norberg, Hans Schumacher and Stefan Jaschke. Moreover, the authors would also like to thank Tobias Herwig, Deepak Pandey and Christian Br\"{u}nger for helpful discussions and comments.


\renewcommand*\appendixpagename{Appendix}
\renewcommand*\appendixtocname{Appendix} 

\appendix
\appendixpage

\section{Proofs} \label{proofs}
\begin{proof}[Proof of Lemma ~\ref{lemma conv matrix and smallest eigenvalue}]
We have
\begin{align*}
& \Big|\Big| \frac{1}{N}  \left(\bm{E}_{\mathcal{K}(N)}^{{\text{\textit{lat}}}}\right)^T\bm{E}_{\mathcal{K}(N)}^{{\text{\textit{lat}}}}-I_{\mathcal{K}(N)}\Big|\Big|_F^2 \\
& = \sum_{j=1}^{\mathcal{K}(N)} \sum_{\ell=1}^{\mathcal{K}(n)} \left(\frac{1}{N} \sum_{n=1}^N e_j^{{\text{\textit{lat}}}}(A_T(z_n))e_{\ell}^{{\text{\textit{lat}}}}(A_T(z_n))-\EE_{\tilde{\PP}}\left[e_j^{{\text{\textit{lat}}}}(A_T(z_n))e_{\ell}^{{\text{\textit{lat}}}}(A_T(z_n))\right]\right)^2\\
\end{align*}
Therefore,
\begin{eqnarray*}
 \EE_{\tilde{\PP}}\left[\Big|\Big| \frac{1}{N}  \left(\bm{E}_{\mathcal{K}(N)}^{{\text{\textit{lat}}}}\right)^T\bm{E}_{\mathcal{K}(N)}^{{\text{\textit{lat}}}}-I_{\mathcal{K}(N)}\Big|\Big|_F^2\right]
& = &\frac{1}{N} \sum_{j=1}^{\mathcal{K}(N)} \sum_{\ell=1}^{\mathcal{K}(N)}   \vari_{\tilde{\PP}}\left[e_j^{{\text{\textit{lat}}}}(A_T(Z))e_{\ell}^{{\text{\textit{lat}}}}(A_T(Z))\right] \nonumber \\
& \leq &\frac{1}{N} \sum_{j=1}^{\mathcal{K}(N)} \sum_{\ell=1}^{\mathcal{K}(N)} \EE_{\tilde{\PP}}\left[\left(e_j^{{\text{\textit{lat}}}}(A_T(Z))e_{\ell}^{{\text{\textit{lat}}}}(A_T(Z))\right)^2\right] \nonumber \\
& = & o(1).
\end{eqnarray*}
Now, (\ref{eq matrix conv}) follows by Markov's inequality.
Since $I_{\mathcal{K}(N)}$ is the identity matrix we have:
$$\lambda_{\min}\left(\frac{1}{N}\left(\bm{E}_{\mathcal{K}(N)}^{{\text{\textit{lat}}}}\right)^T\bm{E}_{\mathcal{K}(N)}^{{\text{\textit{lat}}}}\right)-1=
\lambda_{\min}\left(\frac{1}{N} \left(\bm{E}_{\mathcal{K}(N)}^{{\text{\textit{lat}}}}\right)^T\bm{E}_{\mathcal{K}(N)}^{{\text{\textit{lat}}}}-I_{\mathcal{K}(N)}\right).
$$
The result now follows from the fact that the smallest eigenvalue of a matrix is bounded above by its Frobenius norm and that therefore (\ref{eq matrix conv}) implies (\ref{eq eigenvalue conv}).
\end{proof}


\begin{proof}[Proof of Lemma ~\ref{lemma conv quadratic form parameter}]
By the standard representation of the empirical error $\left( {\hat{\boldsymbol{\alpha}}}_K^{\text{\textit{lat}}}- \boldsymbol{\alpha}_K^{\text{\textit{lat}}} \right)$ for least squares estimators it follows that
\begin{eqnarray*}
\left( {\hat{\boldsymbol{\alpha}}}_K^{\text{\textit{lat}}}- \boldsymbol{\alpha}_K^{\text{\textit{lat}}} \right) = \left(\left(\bm{E}_K^{{\text{\textit{lat}}}}\right)^T\bm{E}_K^{{\text{\textit{lat}}}}\right)^{-1} \left(\bm{E}_K^{{\text{\textit{lat}}}}\right)^T \bm{a}_T^K.
\end{eqnarray*}
Putting $\bm{B}_{\mathcal{K}(N)}= \left((1 \slash N)\left(\bm{E}_{\mathcal{K}(N)}^{{\text{\textit{lat}}}}\right)^T\bm{E}_{\mathcal{K}(N)}^{{\text{\textit{lat}}}}\right)$ we have by the above representation for the empirical error
\begin{align*}
  \hat{\boldsymbol{\alpha}}_{\mathcal{K}(N)}^{\text{\textit{lat}}}-  \boldsymbol{\alpha}_{\mathcal{K}(N)}^{\text{\textit{lat}}} = \bm{B}_{\mathcal{K}(N)}^{-1} \frac{1}{N} \left(\bm{E}_{\mathcal{K}(N)}^{{\text{\textit{lat}}}}\right)^T\bm{a}_T^{\mathcal{K}(N)}.
  \end{align*}
Then
\begin{align}\label{eq proof conv quadratic form}
&\left( \hat{\boldsymbol{\alpha}}_{\mathcal{K}(N)}^{\text{\textit{lat}}}-  \boldsymbol{\alpha}_{\mathcal{K}(N)}^{\text{\textit{lat}}}\right)^T
\left( \hat{\boldsymbol{\alpha}}_{\mathcal{K}(N)}^{\text{\textit{lat}}}-  \boldsymbol{\alpha}_{\mathcal{K}(N)}^{\text{\textit{lat}}}\right) \nonumber \\
& = \frac{1}{N^2} \left(\bm{a}_T^{\mathcal{K}(N)}\right)^T \bm{E}_{\mathcal{K}(N)}^{{\text{\textit{lat}}}} \bm{B}_{\mathcal{K}(N)}^{-1} \bm{B}_{\mathcal{K}(N)}^{-1}
\left(\bm{E}_{\mathcal{K}(N)}^{{\text{\textit{lat}}}}\right)^T\bm{a}_T^{\mathcal{K}(n)} \nonumber \\
& \leq \frac{1}{N^2} \left(\lambda_{\max}\left(\bm{B}_{\mathcal{K}(N)}^{-1}\right)\right)^2 \left(\bm{a}_T^{\mathcal{K}(N)}\right)^T \bm{E}_{\mathcal{K}(N)}^{{\text{\textit{lat}}}}
\left(\bm{E}_{\mathcal{K}(N)}^{{\text{\textit{lat}}}}\right)^T\bm{a}_T^{\mathcal{K}(N)},
\end{align}
where $\lambda_{\max}(\bm{A})$ denotes the largest eigenvalue of a matrix $\bm{A}$.
Notice that by Assumption \ref{as:Newey_1}
\begin{align}\label{eq lemma conv quadratic form cauchy}
&\frac{1}{N^2} \, \EE_{\tilde{\PP}} \left[\left(\bm{a}_T^{\mathcal{K}(N)}\right)^T  \bm{E}_{\mathcal{K}(N)}^{{\text{\textit{lat}}}} \left(\bm{E}_{\mathcal{K}(N)}^{{\text{\textit{lat}}}}\right)^T\bm{a}_T^{\mathcal{K}(N)}\right] \nonumber \\
 & = \frac{1}{N}\, \EE_{\tilde{\PP}}\left[\left(a_T^{\mathcal{K}(N) } \left( A_T(Z) \right) \right)^2 \left(\bm{e}_{\mathcal{K}(N)}^{{\text{\textit{lat}}}}\right)^T \bm{e}_{\mathcal{K}(N)}^{{\text{\textit{lat}}}}\right] \nonumber \\
& \leq  \frac{1}{N}\, \sqrt{\EE_{\tilde{\PP}}\left[\left(a_T^{\mathcal{K}(N)} \left( A_T(Z) \right)\right)^4\right]} \sqrt{\EE_{\tilde{\PP}}\left[\left(\left(\bm{e}_{\mathcal{K}(N)}^{{\text{\textit{lat}}}}\right)^T \bm{e}_{\mathcal{K}(N)}^{{\text{\textit{lat}}}}\right)^2\right]}, \nonumber \\
\end{align}
where we used the Cauchy-Schwarz inequality. Using Assumption \ref{as:Newey_3 modified} and Condition (\ref{th:Newey eq rate}) we get from (\ref{eq lemma conv quadratic form cauchy})
 $$\EE_{\tilde{\PP}}\left[\left(\bm{a}_T^{\mathcal{K}(N)}\right)^T \bm{E}_{\mathcal{K}(N)}^{{\text{\textit{lat}}}}
\left(\bm{E}_{\mathcal{K}(N)}^{{\text{\textit{lat}}}}\right)^T\bm{a}_T^{\mathcal{K}(N)}\right]=o\left(\mathcal{K}(N)^{-\gamma_{{\text{\textit{lat}}}}}\right).$$
By Markov's inequality it follows
$$ \frac{1}{N^2}\left(\bm{a}_T^{\mathcal{K}(N)}\right)^T \bm{E}_{\mathcal{K}(N)}^{{\text{\textit{lat}}}}
\left(\bm{E}_{\mathcal{K}(N)}^{{\text{\textit{lat}}}}\right)^T\bm{a}_T^{\mathcal{K}(N)}=o_{\tilde{\PP}}\left(\mathcal{K}(N)^{-\gamma_{{\text{\textit{lat}}}}}\right).$$
Since $\lambda_{\max}\left(\bm{B}_{\mathcal{K}(N)}^{-1}\right)=\left(\lambda_{\min}\left(\bm{B}_{\mathcal{K}(N)}\right)\right)^{-1}$ Equation (\ref{eq eigenvalue conv}) implies that $\lambda_{\max}\left(\bm{B}_{\mathcal{K}(N)}^{-1}\right)=O_{\tilde{\PP}}(1)$. Putting everything together we get that (\ref{eq proof conv quadratic form}) is indeed $o_{\tilde{\PP}}\left(\mathcal{K}(N)^{-\gamma_{{\text{\textit{lat}}}}}\right)$.
\end{proof}


\begin{proof}[Proof of Lemma ~\ref{lemma conv quadratic form parameter regress now}]
By the standard representation of the empirical error $\left( {\hat{\boldsymbol{\alpha}}}_K^{\text{\textit{now}}}- \boldsymbol{\alpha}_K^{\text{\textit{now}}} \right)$ for least squares estimators it follows that
\begin{eqnarray*}
\left( {\hat{\boldsymbol{\alpha}}}_K^{\text{\textit{now}}}- \boldsymbol{\alpha}_K^{\text{\textit{now}}} \right) = \left(\left(\bm{E}_K^{{\text{\textit{now}}}}\right)^T\bm{E}_K^{{\text{\textit{now}}}}\right)^{-1} \left(\bm{E}_K^{{\text{\textit{now}}}}\right)^T (\bm{a}_{0,t}^K+\bm{p}_{0,t}).
\end{eqnarray*}
Putting $\bm{B}_{\mathcal{K}(N)}= \left((1 \slash
N)\left(\bm{E}_{\mathcal{K}(N)}^{{\text{\textit{now}}}}\right)^T\bm{E}_{\mathcal{K}(N)}^{{\text{\textit{now}}}}\right)$
we have again by the above representation for the empirical error
\begin{align*}
  \hat{\boldsymbol{\alpha}}_{\mathcal{K}(N)}^{\text{\textit{now}}}-  \boldsymbol{\alpha}_{\mathcal{K}(N)}^{\text{\textit{now}}}
  = \bm{B}_{\mathcal{K}(N)}^{-1} \frac{1}{N} \left(\bm{E}_{\mathcal{K}(N)}^{{\text{\textit{now}}}}\right)^T\left(\bm{p}_{0,t}+\bm{a}_{0,t}^{\mathcal{K}(N)}\right).
  \end{align*}
Then
\begin{align*}
&\left( \hat{\boldsymbol{\alpha}}_{\mathcal{K}(N)}^{\text{\textit{now}}}-
\boldsymbol{\alpha}_{\mathcal{K}(N)}^{\text{\textit{now}}} \right)^T
\left( \hat{\boldsymbol{\alpha}}_{\mathcal{K}(N)}^{\text{\textit{now}}}-  \boldsymbol{\alpha}_{\mathcal{K}(N)}^{\text{\textit{now}}}\right) \nonumber \\
& = \frac{1}{N^2}
\left(\bm{a}_{0,t}^{\mathcal{K}(N)}+\bm{p}_{0,t}\right)^T \bm{E}_{\mathcal{K}(N)}^{{\text{\textit{now}}}} \bm{B}_{\mathcal{K}(N)}^{-1}
 \bm{B}_{\mathcal{K}(N)}^{-1}
\left( \bm{E}_{\mathcal{K}(N)}^{{\text{\textit{now}}}} \right)^T\left(\bm{a}_{0,t}^{\mathcal{K}(N)}+\bm{p}_{0,t}\right) \nonumber \\
& \leq \frac{1}{N^2}
\left(\lambda_{max}\left(\bm{B}_{\mathcal{K}(N)}^{-1}\right)\right)^2
\left(\bm{a}_{0,t}^{\mathcal{K}(N)}+\bm{p}_{0,t}\right)^T\bm{E}_{\mathcal{K}(N)}^{{\text{\textit{now}}}}
(\bm{E}_{\mathcal{K}(N)}^{{\text{\textit{now}}}})^T\left(\bm{a}_{0,t}^{\mathcal{K}(N)}+\bm{p}_{0,t}\right),
\end{align*}
where $\lambda_{max}(\bm{A})$ denotes the largest eigenvalue of a matrix $\bm{A}$.
By Assumption \ref{as:Newey_1 regress-now} we have
\begin{align}\label{eq lemma regress-now a + p}
&\frac{1}{N^2}\, \EE_{\tilde{\PP}} \left[ \left(\bm{a}_{0,t}^{\mathcal{K}(N)}+\bm{p}_{0,t}\right)^T\bm{E}_{\mathcal{K}(N)}^{{\text{\textit{now}}}}
(\bm{E}_{\mathcal{K}(N)}^{{\text{\textit{now}}}})^T\left(\bm{a}_{0,t}^{\mathcal{K}(N)}+\bm{p}_{0,t}\right)\right] \nonumber \\
&=\frac{1}{N} \, \EE_{\tilde{\PP}}\left[\left(a_{0,t}^{\mathcal{K}(N)}(A_t(Z))+p_{0,t} (A_T(Z))\right)^2 \left(\bm{e}_{\mathcal{K}(N)}^{{\text{\textit{now}}}}\right)^T \bm{e}_{\mathcal{K}(N)}^{{\text{\textit{now}}}}\right].
 \end{align}
 Now, notice that
 \begin{eqnarray*}
 \EE_{\tilde{\PP}}\left[\left(p_{0,t}(A_T(Z)) \right)^2 \left(\bm{e}_{\mathcal{K}(N)}^{{\text{\textit{now}}}}\right)^T \bm{e}_{\mathcal{K}(N)}^{{\text{\textit{now}}}}\right] & = &
 \EE_{\tilde{\PP}}\left[\EE_{\tilde{\PP}}\left[\left(p_{0,t}(A_T(Z)) \right)^2 \left(\bm{e}_{\mathcal{K}(N)}^{{\text{\textit{now}}}}\right)^T \bm{e}_{\mathcal{K}(N)}^{{\text{\textit{now}}}}\big | A_t(Z)\right]\right]\\
& = & \EE_{\tilde{\PP}}\left[\EE_{\tilde{\PP}}\left[\left(p_{0,t}(A_T(Z)) \right)^2 \big | A_t(Z)\right]\left(\bm{e}_{\mathcal{K}(N)}^{{\text{\textit{now}}}}\right)^T \bm{e}_{\mathcal{K}(N)}^{{\text{\textit{now}}}}\right] \\
& = &\sigma^2 \mathcal{K}(N).
\end{eqnarray*}
 Moreover, since $\EE_{\tilde{\PP}}[p_{0,t}(A_T(Z))|A_t(Z)]=0$ and since $a_{0,t}^{\mathcal{K}(N)}= \sum_{\ell=\mathcal{K}(N)+1}^{\infty} \alpha_{\ell}^{\text{\textit{now}}}
e_{\ell}^{\text{\textit{now}}}(A_t(Z))$ implying that $\EE_{\tilde{\PP}}[a_{0,t}^{\mathcal{K}(N)}(A_t(Z))|A_t(Z)]= a_{0,t}^{\mathcal{K}(N)}(A_t(Z))$,  we obtain
\begin{align*}
&\EE_{\tilde{\PP}} \left[a_{0,t}^{\mathcal{K}(N)}(A_t(Z)) p_{0,t}(A_T(Z)) \left(\bm{e}_{\mathcal{K}(N)}^{{\text{\textit{now}}}}\right)^T \bm{e}_{\mathcal{K}(N)}^{{\text{\textit{now}}}}\right] \\
&=  \EE_{\tilde{\PP}}\left[\EE_{\tilde{\PP}}\left[p_{0,t}(A_T(Z))|A_t(Z)\right]a_{0,t}^{\mathcal{K}(N)}(A_t(Z)) \left(\bm{e}_{\mathcal{K}(N)}^{{\text{\textit{now}}}}\right)^T \bm{e}_{\mathcal{K}(N)}^{{\text{\textit{now}}}}\right] \\
& =0.
\end{align*}
Hence,
\begin{eqnarray*}\label{eq lemma conv quadratic form cauchy regress now}
(\ref{eq lemma regress-now a + p}) & = & \frac{1}{N}\, \EE_{\tilde{\PP}}\left[\left(a_{0,t}^{\mathcal{K}(N)}(A_t(Z))\right)^2 \left(\bm{e}_{\mathcal{K}(N)}^{{\text{\textit{now}}}}\right)^T \bm{e}_{\mathcal{K}(N)}^{{\text{\textit{now}}}}\right] + \frac{\sigma^2 \mathcal{K}(N)}{N} \nonumber \\
& \leq & \frac{1}{N}\, \sqrt{\EE_{\tilde{\PP}}\left[\left(a_{0,t}^{\mathcal{K}(N)}(A_t(Z)) \right)^4\right]}\sqrt{\EE_{\tilde{\PP}}\left[\left(\left(\bm{e}_{\mathcal{K}(N)}^{{\text{\textit{now}}}}\right)^T \bm{e}_{\mathcal{K}(N)}^{{\text{\textit{now}}}}\right)^2\right]} + \frac{\sigma^2 \mathcal{K}(N)}{N}
\end{eqnarray*}
where we used the Cauchy-Schwarz inequality. Using Assumption
\ref{as:Newey_3 modified regress-now} and (\ref{th:Newey eq rate regress now}) we have that (\ref{eq lemma regress-now a + p}) is $o\left(\mathcal{K}(N)^{-\gamma_{\text{\textit{now}}}}\right)+ O(\mathcal{K}(N) \slash N)$. The remaining steps are now as in the proof of Theorem \ref{th:Newey_1}.
\end{proof}

\begin{proof}[Proof of Lemma ~\ref{lem:approx_error_piecewise}]
Let $f$ be the density on $D$. Then $m:=\min_{u \in D} f(u) >0$ and $M:=\max_{u \in D} f(u) < \infty$. We approximate the coefficients $\alpha_{0k}$ and $\alpha_{1k}$ by $g_T(c_k) \slash \sqrt{K}$ and  $g_T'(c_k) \slash C_{1k}$, respectively. By a first order Taylor expansion around $c_k$ with Lagrange's form of the remainder term, i.e.
$$g_T(u) = g_T(c_k) + g_T'(c_k)(u-c_k) + \frac{1}{2} g_T''(\xi) (u-c_k)^2, \; \xi \in [u,c_k],$$
we obtain
\begin{eqnarray*}
\sqrt{K} \alpha_{0k} &=& K \; \EE_{\tilde{\PP}}\left[ g_T(A_T(Z)) \Ind_k(A_T(Z)) \right]  \\
                     &=& K \int_{b_k}^{b_{k+1}} \left( g_T(c_k) + g_T'(c_k)(u-c_k) + \frac{1}{2}g_T''(\xi)(u-c_k)^2 \right) f(u) \, \intd u \\
                     &=& g_T(c_k) + \frac{K}{2} \int_{b_k}^{b_{k+1}} g_T''(\xi) (u-c_k)^2 f(u) \, \intd u,
\end{eqnarray*}
and

\begin{eqnarray*}
C_{1k} \,\alpha_{1k} &=& C_{1k}^2 \, \EE_{\tilde{\PP}} \left[ g_T(A_T(Z)) (A_T(Z)-c_k) \Ind_k(A_T(Z))\right] \\
                   &=& C_{1k}^2 \int_{b_k}^{b_{k+1}} \left( g_T(c_k) + g_T'(c_k)(u-c_k) + \frac{1}{2} g_T''(\xi) (u-c_k)^2\right) (u-c_k) f(u)\, \intd u \\
                   &=& g_T'(c_k) + \frac{C_{1k}^2}{2} \int_{b_k}^{b_{k+1}} g_T''(\xi) (u-c_k)^3 f(u)\, \intd u.
\end{eqnarray*}
The following bounds will be helpful in the remainder of the proof
\begin{eqnarray}\label{first bound proof lemma 4.1}
 \frac{1}{K M}\leq (b_{k+1}-b_k) \leq \frac{1}{K m}.
 \end{eqnarray}
They follow from the fact that by definition $ 1 \slash K = \int_{b_k}^{b_{k+1}} f(u) \intd u$ and the trivial inequalities
$ m (b_{k+1}-b_k) \leq \int_{b_k}^{b_{k+1}} f(u) \intd u \leq M (b_{k+1}-b_k)$.
Moreover
\begin{eqnarray}\label{second bound proof of lemma 4.1}
\max_{1 \leq k \leq K} C_{1k}^2 & = & \frac{1}{\min_{1 \leq k \leq K}\EE_{\tilde{\PP}}[\Ind_k(A_T(Z)) (A_T(Z)-c_k)^2]} \nonumber \\
 & \leq & \frac{12}{m(b_{k+1}-b_k)^3} \nonumber \\
 & \leq & \frac{12 (KM)^3}{m},
\end{eqnarray}
where the second inequality follows from (\ref{first bound proof lemma 4.1}) and the first inequality from the fact
\begin{eqnarray*}
\EE_{\tilde{\PP}}[\Ind_k(A_T(Z)) (A_T(Z)-c_k)^2] & = & \int_{b_k}^{b_{k+1}} (u-c_k)^2 f(u)\, \intd u \geq m \int_{b_k}^{b_{k+1}} (u-c_k)^2\,\intd u\\
&  = & \frac{m}{3}[ (b_{k+1}-c_k)^3-(b_k-c_k)^3]\\
& \geq & \frac{m}{12} (b_{k+1}-b_k)^3,
\end{eqnarray*}
because $(b_{k+1}-c_k)^3-(b_k-c_k)^3$ as a function of $c_k$ is minimized at $c_k=(b_{k+1}+b_k) \slash 2$.\\
For the fourth moment of the approximation error we now obtain with $B:=\sup_{u \in D} \abs{g_T''(u)}$
{\allowdisplaybreaks
\begin{align*}
& \EE_{\tilde{\PP}} \Big[ \big( g_T(A_T(Z))-g_T^K(A_T(Z))\big)^4 \Big]
= \EE_{\tilde{\PP}}\left[ \left( g_T(A_T(Z))- \sum_{k=1}^K \left( \alpha_{0k} e_{0k}(A_T(Z)) + \alpha_{1k} e_{1k}(A_T(Z))\right)\right)^4\right] \\
&= \EE_{\tilde{\PP}}\left[ \left( \sum_{k=1}^K \left( g_T(A_T(Z))-\alpha_{0k} \sqrt{K} - \alpha_{1k} C_{1k} (A_T(Z)-c_k) \right) \Ind_k(A_T(Z)) \right)^4\right] \\
&= \EE_{\tilde{\PP}}\left[ \sum_{k=1}^K \left( g_T(A_T(Z))-\alpha_{0k} \sqrt{K} - \alpha_{1k} C_{1k} (A_T(Z)-c_k) \right)^4 \Ind_k(A_T(Z)) \right] \\
&= \sum_{k=1}^K \int_{b_k}^{b_{k+1}} \Bigg( g_T(u) - g_T(c_k) - g_T'(c_k) (u-c_k) - \frac{K}{2} \int_{b_k}^{b_{k+1}} g_T''(\xi) (v-c_k)^2 f(v) \, \intd v \\
 & \qquad \qquad \qquad - \frac{C_{1k}^2}{2} \left(\int_{b_k}^{b_{k+1}} g_T''(\xi) (v-c_k)^3 f(v)\, \intd v \right) (u-c_k) \Bigg)^4 f(u)\, \intd u \\
&\leq 27 \, \sum_{k=1}^K \Bigg(\int_{b_k}^{b_{k+1}} \left(  \frac{1}{2} g_T''(\xi) (u-c_k)^2 \right)^4 f(u) \, \intd u + \int_{b_k}^{b_{k+1}} \left(\frac{K}{2} \int_{b_k}^{b_{k+1}} g_T''(\xi) (v-c_k)^2 f(v) \, \intd v\right)^4 f(u) \, \intd u \\
& \qquad \qquad \quad + \int_{b_k}^{b_{k+1}} \left( \frac{C_{1k}^2}{2} \left(\int_{b_k}^{b_{k+1}} g_T''(\xi) (v-c_k)^3 f(v)\, \intd v \right) (u-c_k) \right)^4 f(u) \, \intd u \Bigg) \\
&= 27 \sum_{k=1}^K \Bigg( \frac{1}{16} \int_{b_k}^{b_{k+1}} g_T''(\xi)^4 (u-c_k)^8 f(u)\,\intd u + \frac{K^4}{16} \int_{b_k}^{b_{k+1}} \left( \int_{b_k}^{b_{k+1}} g_T''(\xi) (v-c_k)^2 f(v) \, \intd v\right)^4 f(u) \, \intd u \\
& \qquad \qquad + \frac{C_{1k}^8}{16} \int_{b_k}^{b_{k_1}} \left( \int_{b_k}^{b_{k+1}} g_T''(\xi) (v-c_k)^3 f(v) \, \intd v \right)^4 (u-c_k)^4 f(u) \, \intd u \Bigg) \\
&\leq \frac{27}{16}\, K\, \Bigg(\frac{1}{K} B^4 \max_{1 \leq k \leq K} (b_{k+1}-b_k)^8 + K^4  \frac{1}{K^5}B^4 \left( \max_{1 \leq k \leq K} (b_{k+1}-b_k)^2 \right)^4 \\
& \qquad \qquad +\frac{1}{K^5} \max_{1 \leq k \leq K} \left(C_{1k}^8\right) B^4 \left( \max_{1 \leq k \leq K} (b_{k+1}-b_k)^3\right)^4 \max_{1 \leq k \leq K} (b_{k+1}-b_k)^4\Bigg) \\
&\leq \frac{54}{16} \cdot B^4 \cdot \frac{1}{m^8K^8}  + \frac{27}{16} \frac{1}{K^4} \frac{12^4 M^{12}K^{12}}{m^4} \cdot B^4 \cdot \frac{1}{K^{12} m^{12}} \frac{1}{K^4 m^4} \\
&=O \left( \frac{1}{K^8} \right),
\end{align*}
}
where we used that $g_T(c_k)+g_T'(c_k)(u-c_k)$ corresponds to the first order Taylor expansion of $g_T(u)$ around $c_k$. The first inequality follows from Lo\`{e}ve's $\text{c}_r$-inequality and the third makes use of (\ref{first bound proof lemma 4.1}) and (\ref{second bound proof of lemma 4.1}). Lemma \ref{lem:approx_error_piecewise} follows immediately.

\end{proof}


\begin{proof}[Proof of Lemma ~\ref{lem:ratio moments pw}]
Let $m$ and $M$ as in the proof of Lemma \ref{lem:approx_error_piecewise}. Let $b_k$ and $b_{k+1}$ be in $[a_1,a_2]$ with $b_k<b_{k+1}$ and let $c_k \in [b_k,b_{k+1}]$.
Then
$$\int_{b_k}^{b_{k+1}} (u-c_k)^4 f(u)\,\intd u \leq M \int_{b_k}^{b_{k+1}} (u-c_k)^4\,\intd u \leq M (b_{k+1}-b_k)^5.$$
Moreover, from the proof of Lemma \ref{lem:approx_error_piecewise} we know
$$\int_{b_k}^{b_{k+1}} (u-c_k)^2 f(u)\,\intd u \geq \frac{m}{12} (b_{k+1}-b_k)^3.$$
Therefore
\begin{eqnarray*}
\frac{\int_{b_k}^{b_{k+1}} (u-c_k)^4 f(u)\,\intd u}{\left(\int_{b_k}^{b_{k+1}} (u-c_k)^2 f(u)\,\intd u \right)^2} & \leq & \frac{M(b_{k+1}-b_k)^5}{\left(\frac{m}{12}\right)^2 (b_{k+1}-b_k)^6}\\
& = & C \; \frac{1}{(b_{k+1}-b_k)},
\end{eqnarray*}
where $C:=M \slash (m \slash 12)^2$.
Using the left hand inequality in (\ref{first bound proof lemma 4.1}) we get
\begin{eqnarray*}
\frac{\int_{b_k}^{b_{k+1}} (u-c_k)^4 f(u)\,\intd u}{\left(\int_{b_k}^{b_{k+1}} (u-c_k)^2 f(u)\,\intd u \right)^2} \leq C \cdot M \cdot K.
\end{eqnarray*}

\end{proof}


\newpage
\bibliographystyle{apalike}
\bibliography{Bib_Article1}

\end{document}